\documentclass[envcountsame]{sigplanconf}

\pdfoutput=1

\usepackage{twelftom2}

\let\oldmaketitle=\maketitle
\gdef\maketitle{
\oldmaketitle
\thispagestyle{plain}
\pagestyle{plain}
}

\title{Towards Extracting Explicit Proofs from Totality Checking in
  \Twelf\thanks{Draft of \today}}
\authorinfo{Yuting Wang}
           {University of Minnesota, USA}
           {yuting@cs.umn.edu}
\authorinfo{Gopalan Nadathur}
           {University of Minnesota, USA}
           {gopalan@cs.umn.edu}

\begin{document}

\maketitle

\begin{abstract}
The Edinburgh Logical Framework (\LF) is a dependently type
$\lambda$-calculus that can be used to encode formal systems.
The versatility of \LF allows specifications to be constructed also
about the encoded systems.
The \Twelf system exploits the correspondence between formulas and
types to give specifications in \LF a logic programming
interpretation.
By interpreting particular arguments as input and others as output,
specifications can be seen as describing non-deterministic functions.
If particular such functions can be shown to be total, they represent
constructive proofs of meta-theorems of the encoded systems.
\Twelf provides a suite of tools for establishing totality.
However, all the resulting proofs of meta-theorems are implicit: \Twelf's
totality checking does not yield a certificate that can be given to a
proof checker.
We begin the process here of making these proofs explicit.
We treat the restricted situation in \Twelf where context definitions
(regular worlds) and lemmas are not used.
In this setting we describe and prove correct a translation of the
steps in totality checking into an actual proof in the companion logic
$\M2$.
We intend in the long term to extend our translation to all of \Twelf
and to use this work as the basis for producing proofs in the related
Abella system.

\end{abstract}

\section{Introduction}

The Edinburgh Logical Framework (\LF) is a general framework for
formalizing systems that are presented in a rule-based and
syntax-directed fashion \cite{harper87lics}.
\LF is based on a dependently-typed \lcalculus and, as such,
supports the higher-order abstract syntax approach to representing
binding structure. 
Because of its features, \LF has proven to be extremely versatile in
encoding formal systems such as programming languages and logics.

One of the purposes for specifying a formal system is to be
able to prove properties about it; such properties are usually called
\emph{meta-theorems}. 
There are two approaches to doing this relative to \LF
specifications. 
One of these approaches is implemented by the \Twelf system 
that is based on according a logic programming interpretation to \LF
specifications \cite{rohwedder96esop,pfenning91lf}. 
In essence, \Twelf interprets types as relations and thereby
transforms the question of validity of a relation into one about the 
inhabitation of a type. 
\Twelf also allows such relations to be moded, \ie, it lets some
arguments to be designated as inputs and others to be identified as
outputs. 
Relations thus define non-deterministic functions
from inputs to outputs.
\Twelf complements such a treatment with tools for determining if
particular relational specifications represent total functions from
ground inputs to ground outputs.  
When the relations are about types that specify particular aspects of
formal systems, something that is possible to do in the rich types
language of \LF, such verifications of totality correspond to
implicit, constructive proofs of meta-theorems.

The second approach is to proving meta-theorems is to do this
explicitly within a suitable logic. 
The \M2 logic has been described for the specific purpose of
constructing such proofs over \LF specifications
\cite{schurmann98cade,schurmann00phd}. 
Meta-theorems are represented by \M2 formulas and proved by applying
the derivation rules in \M2.
A successful proof returns a certificate in the form of a proof term
that can be checked independently by another (simpler) program. 

Of the two approaches, totality checking has proven to be vastly more
popular: a large number of verifications based on it have been
described in the literature. 
By contrast, although an automatic theorem prover based on \M2 has
been developed \cite{schurmann00phd}, it has, to our knowledge, not
been used in many reasoning tasks. 
However, there is a potential drawback to using totality checking in
verification: this process does not yield an outcome that can be
provided to a third party that can, for instance, easily check its
correctness. 

This paper begins an effort to address the abovementioned shortcoming
by developing a method for extracting an explicit proof from totality
checking. 
In this first step, we limit ourselves to a subset of the kinds of
verification treated by Twelf.
In particular, we do not consider meta-theorems that require the use
of contexts (defined via regular worlds in Twelf) and also disallow
the use of lemmas in specifications. 
In this setting, we describe and prove correct a procedure for
obtaining a proof in \M2 from the work done by \Twelf towards
establishing totality.

The rest of the paper is organized as follows.
Section \ref{sec:lf} recalls \LF and its interpretation in
\Twelf and introduces some associated terminology needed in the
paper. 
Section \ref{sec:total} describes totality checking. 
The components that make it up have been presented in different
settings (\eg, see
\cite{rohwedder96esop,schurmann03tphols,pfenning02guide}) and some
aspects (such as output coverage checking) have not been described in
the formality needed for what we do in this paper.
Thus, this section represents our attempt to describe these aspects
with in a coordinated and precise fashion. 
Section \ref{sec:m2} presents the \M2 logic. 
Section \ref{sec:genproofs} finally tackles the generation of \M2
proofs from totality checking.
Section \ref{sec:conc} concludes the paper with a brief discussion of
future directions to this work.

\section{The Edinburgh Logical Framework}\label{sec:lf}

\newcommand\lfOAbsId[1]{\lfOAbs {#1} {\ktm} {#1}}
\newcommand\ecdAbsId[1]{\ecdAbs {(\lfOAbsId {#1})}}
\newcommand\ecdAbsXX{\ecdAbsId x}
\newcommand\ecdAbsYY{\ecdAbsId y}
\newcommand\ecdAppAbsXAbsY{\ecdApp {(\ecdAbsXX)} {(\ecdAbsYY)}}

\long\def\lfrules{
  \begin{figure*}[!t]
    Valid Families:
    \begin{smallgather}
      \infer{
        \lfSeq {\G } {\Si} {a} {K}
      }{
        a:K  \in \Si
        &
        \lfCtx {\Si} {\G}
      }
      \qquad
      \infer[]{
        \lfSeq {\G } {\Si} {\lfTAbs x A B } {\ktype}
      }{
        \lfSeq {\G, x:A } {\Si} {B } {\ktype}
      }
    \end{smallgather}
    \begin{smallgather}
      \infer{
        \lfSeq {\G } {\Si} {\lfOAbs x A B } {\lfTAbs x A K}
      }{
        \lfSeq {\G, x:A } {\Si} {B } {K}
      }
      \qquad
      \infer{
        \lfSeq {\G } {\Si} {A\,M } {K[M/x]}
      }{
        \lfSeq {\G } {\Si} {A } {\lfTAbs x B K}
        &
        \lfSeq {\G } {\Si} {M } {B}
      }
    \qquad
      \infer{
        \lfSeq {\G } {\Si} {A } {K^\prime}
      }{
        \lfSeq {\G } {\Si} {A } {K}
        &
        \lfKind {\G} {\Si} {K^\prime}
        &
        \lfDefEq {\G} {\Si} K {K^\prime}
      }
    \end{smallgather}
\\
    Valid Objects:
    \begin{smallgather}
      \infer{
        \lfSeq {\G } {\Si} {c } {A}
      }{
        c:K \in \Si
        &
        \lfCtx {\Si} {\G}
      }
      \qquad
      \infer{
        \lfSeq {\G } {\Si} {x } {A}
      }{
        x:A \in \G
        &
        \lfCtx {\Si} {\G}
      }
    \end{smallgather}
    \begin{smallgather}
      \infer{
        \lfSeq {\G } {\Si} {\lfOAbs x A M } {\lfTAbs x A B}
      }{
        \lfSeq {\G, x:A } {\Si} {M } {B}
      }
      \qquad
      \infer{
        \lfSeq {\G } {\Si} {M\,N } {B[N/x]}
      }{
        \lfSeq {\G } {\Si} {M } {\lfTAbs x A B}
        &
        \lfSeq {\G } {\Si} {N} {A}
      }
    \qquad
      \infer{
        \lfSeq {\G } {\Si} {M } {A^\prime}
      }{
        \lfSeq {\G } {\Si} {M } {A}
        &
        \lfSeq {\G } {\Si} {A^\prime } {\ktype}
        &
        \lfDefEq {\G} {\Si} A {A^\prime}
      }
    \end{smallgather}
\nocaptionrule \caption{Rules for Valid Type Families and Objects}
  \label{fig:lfrules}
  \end{figure*}
}

We provide a brief summary of the Edinburgh Logical Framework (\LF) in
this section and explain its use in encoding the syntax and derivation
rules of formal systems.
We then introduce the \Twelf system that provides \LF specifications a
logic programming interpretation, which becomes a basis for proving
meta-theoretic properties of systems encoded in \LF.
We conclude the section with the definitions of a few technical
notions like substitution and unification related to \LF that will be
needed in later parts of this paper.

\subsection{Syntax and Typing Judgments}

The syntax of \LF is that of a $\lambda$-calculus that has three
categories of expressions: kinds, type families and objects.
Kinds classify type families and type families classify objects.
The special kind for types is $\ktype$.
We use the symbol $a$ for type level constants, $c$ for object level
constants and $x$ for object-level variables.
The different categories of expressions are given by the rules below;
we use $K$, possibly with subscripts, as a schematic
variable for kinds, $A$ and $B$ for type families, $M$
and $N$ for objects, and $U$ and $V$ for either type families or objects.
\begin{align*}
  \text{Kinds}\qquad    &       K  ::= \ktype \sep \lfTAbs x A K\\
  \text{Type Families}\qquad  &  A, B  ::= a \sep \lfTAbs x A B \sep \lfApp A M \\
  \text{Objects}\qquad     &    M, N  ::= c \sep x \sep \lfOAbs x A M \sep \lfApp M N
\end{align*}
Here, $\lfOBnd$ and $\Pi$ are binding operators.
\LF is dependently typed in the sense that kinds and type
families might depend on object terms:
we have kinds of the form $\lfTAbs{x}{A}{K}$ and type
families of the form $\lfTAbs x A B$, where $\lfTBnd$ binds
occurrences of the (object) variable $x$ of type $A$ in the kind
$K$  and type family $B$, respectively.
We write $\lfTAbsa {x_1:A_1, \ldots, x_n:A_n} B$ as an abbreviation for
$\lfTAbs {x_1} {A_1} {\ldots \lfTAbs{x_n} {A_n} B}$ and similarly for
$\lfTAbs {x_1} {A_1} {\ldots \lfTAbs{x_n} {A_n} K}$.
The expression $\lfTAbs x A B$ will also be written as $A \to B$ if
the variable $x$ does not occur in $B$.
We take $\to$ to be a right associative operator.

\LF terms are constructed relative to \emph{signatures}, denoted
schematically by $\Si$, that keep track of types
and kinds assigned to constants, and \emph{contexts}, denoted
schematically by $\G$, that keep track of types assigned to variables.
Signatures and contexts are given by the following rules:
\begin{align*}
  \text{Signatures}\qquad \Si & ::= \emptyCol \sep \Si, a:K \sep \Si, c:A\\
  \text{Contexts}\qquad   \G & ::= \emptyCol \sep \G, x:A
\end{align*}
Here, $\emptyCol$ denotes the empty sequence.
Since $\lfTAbsa {x_1:A_1, \ldots, x_n:A_n} B$ has list of binders
that look like an \LF context, we will sometimes write it as $\lfTAbsa
{\G} B$ where $\G = x_1:A_1,\ldots,x_n:A_n$.

\lfrules

\LF is equipped with rules for deriving the following
judgments:
$\lfDefEq {\G} {\Si} M N$ and $\lfDefEq {\G} {\Si} A B$
for equality between objects and types;
$\lfSig {\Si}$, $\lfCtx {\Si} {\G}$ and $\lfKind {\G} {\Si} K$
representing the validity of signatures, contexts and kinds;
and $\lfSeq {\G} {\Si} A K$ and $\lfSeq {\G} {\Si} {M} {A}$
asserting the well-typedness of type families and objects.
Fig.\ref{fig:lfrules} shows the rules for the last two judgments;
the full collection can be found in \cite{harper93jacm}.
We will sometimes leave the signature implicit in
these judgments, writing them simply as $\lfSeq {\G } {} {A
} {K}$ and $\G \lfSeq {} {} {M } {A}$.
The notation $U[V/x]$ used in the rules stands for the capture
avoiding substitution of $V$ for the variable $x$ in $U$.
The equality relation $\equiv$ between type families and objects
corresponds to $\beta\eta$-conversion.
The \LF type theory guarantees that every well-typed \LF term has an
unique $\beta\eta$-long normal form called its \emph{canonical form}.
Thus, two well-typed \LF terms are equal iff their canonical
forms are identical up to a renaming of bound variables.
We will often confuse a well-typed \LF term with its
canonical form.

\subsection{\LF as a Specification Language}

\LF can be used to encode varied formal systems via the \emph{types as
  judgments}, \emph{objects as proofs} and \emph{inhabitation as
  provability} principles.
We illustrate this aspect by considering the encoding of the untyped
\lcalculus; this example will be used also in later parts of the
paper.

The object system syntax is given by the following rule:
\begin{align*}
   M ::= \lcAbs x M \sep \lcApp {M_1} {M_2}
\end{align*}
We shall represent these terms using the \LF type \ktm and
object-level constants $\kapp: \ktm \to \ktm \to \ktm$ and $\kabs:
(\ktm \to \ktm) \to \ktm$.
The  precise encoding $\encode{\cdot}$ is given as follows:
\begin{smallalign}
\encode{x} = x
\quad
\encode{\lcApp M N} = \ecdApp {\encode M} {\encode N}
\quad
\encode{\lcAbs x M} = \ecdAbs {(\lfOAbs {x} {\ktm} {\encode M})}
\end{smallalign}
Note that we are representing binding in the \lcalculus here by
binding in \LF.
For example, $\lcAbs x x$ is encoded as an LF term $\lfApp {\kabs} {(\lfOAbs
x {\ktm} x)}$.

Object system judgments are typically represented in LF via types.
Thus, consider the call-by-name evaluation relation between
the \lterms that is defined by the following rules:
\begin{smallgather}
  \infer[\kwd{evAbs}]{
    \lcEval {\lcAbs x M} {\lcAbs x M}
  }{
  }
\qquad
  \infer[\kwd{evApp}]{
\lcEval {\lcApp M N} {V}
  }{
\lcEval {M} {\lcAbs x {M^\prime}}
    &
\lcEval {M^\prime[N/x]} {V}
  }
\end{smallgather}
This relation can be represented by the \LF
type family $\keval : \ktm \to \ktm \to \ktype$: the judgment $\lcEval
M N$ would then be encoded as the \LF type $\ecdEval {\encode M}
{\encode N}$.
The derivability of a judgment $\lcEval M N$ then boils down to the
inhabitation of $\ecdEval {\encode M} {\encode N}$ in \LF.
To determine such inhabitation, we introduce the following
object-level constants into the \LF signature:
\begin{smallalign}
  &\kevapp :
  \lfTAbs M {\ktm} {\lfTAbs {M^\prime} {\ktm \to \ktm} {
      \lfTAbs N {\ktm} {\lfTAbs V {\ktm} \\
 & \quad {
    \ecdEval M {(\ecdAbs {M^\prime})}
    \to \ecdEval {(\lfApp {M^\prime} N)} V
    \to \ecdEval {(\ecdApp M N)} V
    }}}}. \\
  &\kevabs :
  \lfTAbs M {\ktm \to \ktm}
     {\ecdEval {(\ecdAbs M)} {(\ecdAbs M)}}.
\end{smallalign}
Given these constants, the object system derivation
\begin{smallgather}
  \infer[\kwd{evApp}]{
\lcEval {(\lam x.x)\, (\lam y.y)} {\lam y.y}
  }{
    \infer[\kwd{evAbs}]{
\lcEval {\lam x.x} {\lam x.x}
    }{}
    &
    \infer[\kwd{evAbs}]{
\lcEval {x[\lam y.y/x]} {\lam y.y}
    }{}
  }
\end{smallgather}
would be represented by the \LF object
\begin{smallalign}
\small
& \ecdEvAppF
    {(\ecdAbsId x)}
    {(\lfOAbsId x)}
    {(\ecdAbsId y)}\\
& \quad
    {(\ecdAbsId y)}
    {(
      \ecdEvAbsF
          {(\lfOAbsId x)}
      )}\,
    {(
      \ecdEvAbsF
          {(\lfOAbsId y)}
      )}.
\end{smallalign}
This expression has the type
\begin{smallalign}
\ecdEval {(\ecdAppAbsXAbsY)}
{(\ecdAbsYY)}
\end{smallalign}
which is the \LF encoding of
$\lcEval {\lcApp {(\lcAbs x x)} {(\lcAbs y y)}} {\lcAbs y y}$.

Suppose that our object system also encompasses typing judgments of
the form $\stlcOf {\D} M T$ where $\D = x_1:T_1,\ldots,x_n:T_n$: such a
judgment asserts that $M$ has type $T$ under the
assumption that, for $1 \leq i \leq n$, $x_i$ has type $T_i$.
Moreover, let these judgments be defined by the rules
\begin{smallgather}
\infer[\kwd{ofAbs}]{
  \stlcOf {\D} {(\lam x. M)} {T_1 \to T_2}
}{
  \stlcOf {\D, x:T_1} M T_2
}
\qquad
\infer[\kwd{ofApp}]{
  \stlcOf {\D} {\lcApp M N} T_2
}{
  \stlcOf {\D} M {T_1 \to T_2}
  &
  \stlcOf {\D} N {T_1}
}
\end{smallgather}
\noindent
with the $\kwd{ofAbs}$ rule having the side condition that $x$ does
not occur in $\D$.
We use the \LF type $\kty$ and the \LF  (object-level) constant $\karr :
\kty \to \kty \to \kty$ to represent the types of the object system:
the $\encode{\emptyCol}$ mapping is extended to \lcalculus types in such a
way that $\encode{T_1 \to T_2} = \ecdArr {\encode{T_1}}
\encode{T_2}$.
The typing rules for the \lcalculus can then be encoded in \LF using
the signature:
\begin{smallalign}
& \kof : \ktm \to \kty \to \ktype.\\
& \kofapp : \lfTAbs M {\ktm} {\lfTAbs N {\ktm}
    {\lfTAbs {T_1} {\kty} {\lfTAbs T {\kty} \\
& \qquad {
      \ecdOf M {(\ecdArr {T_1} T)}
      \to \ecdOf N {T_1}
      \to \ecdOf {(\ecdApp M N)} T
      }}}}.\\
& \kofabs : \lfTAbs M {\ktm} {\lfTAbs {T_1} {\kty} {\lfTAbs T {\kty} \\
& \qquad {
        (\lfTAbs x {\ktm} {\ecdOf x {T_1} \to \ecdOf {(\lfApp M x)} T})
        \to \ecdOf {(\ecdAbs M)} {(\ecdArr {T_1} T)}
      }}}
\end{smallalign}
We assume here that $\stlcOf {\D}  M T$
is encoded as
$\lfSeq {\encode{\D}} {} {\encode{M}} {\encode{T}}$,
where $\encode{\D}$ is an \LF context resulting from transforming every
$x_i : T_i$ in $\D$ into
$x_i:\ktm, y_i: \ecdOf {x_i} {\encode{T_i}}$.
%

\LF expressions contain a lot of verbose type information.
This situation can be eased by making some of the outermost binders in
the types assigned to object level constant implicit by using tokens
beginning with uppercase letters for the variables they bind.
Thus, the constants for encoding the evaluation and
typing rules for the \lcalculus can be shown by means of the
following signature:
\begin{smallalign}
 &\kevapp :
    \ecdEval M {(\ecdAbs {M^\prime})}
    \to \ecdEval {(\lfApp {M^\prime} N)} V
    \to \ecdEval {(\ecdApp M N)} V.\\
  &\kevabs :
    \ecdEval {(\ecdAbs M)} {(\ecdAbs M)}.\\
& \kofapp :
      \ecdOf M {(\ecdArr {T_1} T)}
      \to \ecdOf N {T_1}
      \to \ecdOf {(\ecdApp M N)} T.\\
& \kofabs :
        (\lfTAbs x {\ktm} {\ecdOf x {T_1} \to \ecdOf {(\lfApp M x)} T})
        \to \ecdOf {(\ecdAbs M)} {(\ecdArr {T_1} T)}.
\end{smallalign}
When we make the binding implicit in this way, it will always be
the case that the types for the bound variables can be uniquely inferred
by a \emph{type reconstruction} process
\cite{pientka13jfp}.
When showing terms, we will also leave out the arguments for constants
that correspond the implicitly bound variables, assuming that these
too can be inferred.
Using this convention, the \LF object for the evaluation derivation
considered earlier can be written as follows:
\begin{smallalign}
 \ecdEvApp & {(\ecdEvAbs : \ecdEval {(\ecdAbsXX)} {(\ecdAbsXX)})}\\
           & {(\ecdEvAbs : \ecdEval {(\ecdAbsYY)} {(\ecdAbsYY)})}
\end{smallalign}
The missing arguments of $\kevapp$ and the two occurrences of
$\kevabs$ will be filled in here by type reconstruction.

\subsection{\Twelf and Logic Programming in \LF}

\Twelf is a tool based on \LF that uses the representation principles
just discussed for specifying and reasoning about formal systems.
\Twelf permits $A \from B$ as an alternative syntax for $B \to A$,
treating $\from$ as a left associative operator.
An important idea underlying \Twelf is that \LF types can be given a
logic programming interpretation so that they can be executed.
The full description of the operational semantics of \Twelf can
be found in \cite{pfenning91lf}.
We consider here only the interpretation for constant
definitions of the form
\begin{align*}
  & c : \lfAppFour a {M_1} {\ldots} {M_n}
  \from \lfAppFour a {M_{1_1}} {\ldots} {M_{1_n}}
  \from \ldots
  \from \lfAppFour a {M_{m_1}} {\ldots} {M_{m_n}}
\end{align*}
where $a : \lfTAbsa {x_1:A_1,\ldots,x_n:A_n} {\ktype}$ is an \LF type family.
In the logic programming setting, this definition of $c$ is called a
\emph{clause} that has $\lfAppFour a {M_1} {\ldots} {M_n}$ as its
\emph{head} and $\lfAppFour a {M_{i_1}} {\ldots} {M_{i_n}}(1
\leq i \leq m)$ as its \emph{premises}.
We think of such a clause in the same way we would a \Prolog clause.
Thus, we call the type family $a$ a predicate and we think of the
clause as one \emph{for} $a$.
Given a set of clauses for $a$, we query \Twelf for the solutions of
\emph{goals} of the form
$M: \lfAppFour a {M_1^\prime} {\ldots} {M_n^\prime}$.
\Twelf interprets this as a question that asks if
a well-formed term of type $\lfAppFour a {M_1^\prime} {\ldots}
{M_n^\prime}$ exists, treating $M$ and the free variables in $M_1,
\ldots, M_n$ as logic variables.
\Twelf performs goal-directed proof search using backchaining
similar to \Prolog, using unification to
instantiate  logic variables as needed to
solve the inhabitation question \cite{pfenning91lf}.
We assume that the search will be based on trying the premises in left
to right order.

\subsection{Substitution and Unification in \LF}

A substitution $\si$ in \LF is a type-preserving mapping from
variables to objects that differs from the identity at only finitely
many places.
We write $\si$ as $(M_1/x_1, \ldots, M_n/x_n)$ where $M_i/x_i\,(1
\leq i \leq n)$ are the only non-identity mappings.
For any \LF term $t$, $t[\si]$ is the term obtained by
applying $\si$ in a capture avoiding way to the free variables in $t$.

Substitutions transform terms that are well-typed in one context into
ones that are well-typed in another context, something that is
asserted by the judgment $\lfSeq {\G} {\Si} {\si} {\G^\prime}$ that is
defined by the following rules:
\begin{smallalign}
  \infer[\kwd{subst-ept}]{
    \lfSeq {\G} {\Si} {\emptyCol} {\emptyCol}
  }{
  }
  \quad
  \infer[\kwd{subst-typ}]{
    \lfSeq {\G} {\Si} {(\si, M/x)} {(\Gamma^\prime, x:A)}
  }{
    \lfSeq {\G} {\Si} M {A[\si]}
    &
    \lfSeq {\G} {\Si} {\si} {\Gamma^\prime}
  }
\end{smallalign}
Given a $\sigma$ and a $\G'$, there is a unique $\G$ with the
smallest domain such that $\lfSeq {\G} {\Si} {\si} {\G'}$ is
derivable. We shall intend to pick out this $\G$ when we use this
judgment in later sections.
We will also want to use the judgment $\lfSeq {\G} {\Si} {\si} {\G'}$
when the domain of $\si$ is a subset of that of $\G'$. We will assume
in this case that the domain of $\si$ is extended with identity
substitutions to match that of $\G'$. 
If $\G$ is the context $x_1:A_1, \ldots, x_n:A_n$, we write $\G[\si]$ to
represent the context $x_1:A_1[\si], \ldots, x_n:A_n[\si]$.
The composition of substitutions $\substComp{\si} {\theta}$ is
defined as follows:
\begin{gather*}
  \substComp {\emptyCol} {\theta} = \emptyCol
  \qquad
  \substComp {(\si, M/x)} {\theta} =
  (\substComp {\si} {\theta}, M[\theta]/x)
\end{gather*}

The following lemmas are easily proved by induction on typing
derivations related to substitutions:
\begin{lemma}\label{lem:subst_comp}
If $\lfSeq {\G_2} {\Si} {\si_1} {\G_1}$ and $\lfSeq {\G_3} {\Si}
{\si_2} {\G_2}$ are derivable then so is
$\lfSeq {\G_3} {\Si} {\substComp {\si_1}{\si_2}} {\G_1[\si_2]}$.
\end{lemma}
%

\begin{lemma}
If $\lfSeq{\G} {\Si} {\si} {\G^\prime}$ and $\lfSeq {\G^\prime} {\Si}
{U} {V}$ are derivable then so is $\lfSeq {\G} {\Si} {U[\si]} {V[\si]}$.
\end{lemma}
%

A unification problem $S$ is a finite multiset of
equations $\{t_i = s_i\, |\, 1 \leq i \leq n\}$ where, for $1 \leq i
\leq n$, $t_i$ and $s_i$ are \LF terms of the same kind or type.
A substitution $\si$ such that
$t_i[\si] = s_i[\si]$ for all $i$ such that $1 \leq i \leq n$ is a
unifier for $S$.
It is a most general unifer (\emph{mgu}) if for
any unifer $\theta$ of $S$ there exists a substitution
$\gamma$ such that for any term $t$ we have $t[\theta] =
t[\si][\gamma]$.
Not every unification problem in \LF has an mgu and unifiability is
also not decidable in general.
However, these properties hold when all occurrences of free variables
in the terms determining the unification problem are strict as per the
following definition~\cite{pfenning91lics}.
\begin{definition}
An occurrence of a free variable is strict if it is not in the
argument of a free variable and all its arguments are distinct bound
variables.
\end{definition}
\noindent%
We will also be interested often in matching, \ie, unification where
free variables occur in the terms on only one side of the
equations.
In this case decidability and the existence of most general solutions
follows if every free variable has \emph{at least one} strict
occurrence in the terms.
We refer to a term that satisfies this property as a \emph{strict term}.

Given a context $\G = x_1:A_1, \ldots, x_n:A_n$ we will often write
$\lfApp M \G$ to mean $\lfAppFour M {x_1} {\ldots} {x_n}$.
Similarly, given a substitution $\si = (M_1/x_1, \ldots, M_n/x_n)$ we
will write  $\lfApp M {\si}$ to mean $\lfAppFour M {M_1} {\ldots} {M_n}$.

\section{Totality Checking in Twelf}\label{sec:total}


%
Under the logic programming interpretation of \LF specifications in
\Twelf, a type family $a:\lfTAbsa {\G} {\ktype}$ represents a relation
between its arguments.
This relation can be read as a meta-theorem about the system specified
by interpreting particular arguments as inputs and others as outputs.
For instance, the \emph{subject reduction} theorem for the
\lcalculus (\STLC) states that the evaluation preserves types: for all
terms $E$ and $V$ and types $T$, if $\lcEval E V$ and
$\stlcOf {\emptyCol} E T$ hold, then $\stlcOf {\emptyCol} V T$ holds.
Based on our encoding of evaluation and typing for the \lcalculus, we
can define the following type family:
\begin{align*}
&\ksubred : \ecdEval E V \to \ecdOf E T \to \ecdOf V T \to \ktype.
\end{align*}
By assigning $E$, $V$, $T$ and the first two arguments to \ksubred as
inputs and the third argument to \ksubred as output, the proof search
for $\ecdSubred {D_1} {D_2} {D_3}$ becomes directional:
it queries the existence of a derivation $D_3$ for $\ecdOf V T$
given derivations $D_1$ for $\ecdEval E V$ and $D_2$ for $\ecdOf E T$.
Thus, we can interpret \ksubred operationally as a
non-deterministic function that computes a ground output $D_3$ from
ground inputs $E,V,T,D_1,D_2$.
If we can show that this function is total, \ie, that, given any
ground terms for the inputs, proof search will be able to find a
satisfying ground term for the output in a finite number of steps,
then we would have obtained a constructive proof for the subject
reduction theorem.

\Twelf uses the above approach to interpret meta-theorems as
\emph{totality assertions} and their proofs as \emph{totality
  checking}.
In \Twelf notation,  the subject reduction theorem is expressed as
follows:
\begin{align*}
&\ksubred : \ecdEval E V \to \ecdOf E T \to \ecdOf V T \to \ktype.\\
&\kmode\; \ecdSubred {\inMod{D_1}} {\inMod{D_2}} {\outMod{D_3}}.\\
&\ktotal\; D\; (\ecdSubred D {\_} {\_}).
\end{align*}
The mode declaration that begins with \kmode designates the explicit
arguments prefixed by \inMod{} as inputs and those prefixed by
\outMod{} as outputs; such a designation must be extended to also
include the implicit arguments.
The declaration that begins with \ktotal asserts that
\ksubred represents a total non-deterministic function in the
indicated mode.
This declaration also identifies an argument on which to base a
termination argument as we shall see presently.

\begin{figure}[!t]
\begin{smallalign}
&\kwd{sr-app} :
  \ecdSubred
      {(\ecdEvApp {Dev_1} {Dev_2})}
      {(\ecdOfApp {Dty_1} {Dty_2})}
      {Dty}\\
&\qquad \from
      \ecdSubred
          {(Dev_1 : \ecdEval M {(\ecdAbs {M^\prime})})} {} {}\\
&\qquad\qquad\qquad
          {(Dty_1 : \ecdOf M {(\ecdArr {T_1} T_2)})}\\
&\qquad\qquad\qquad
          {(\ecdOfAbs
            {(Dty_3 : \lfTAbs x {\ktm}
              {\ecdOf x {T_1} \to \ecdOf {(\lfApp {M^\prime} x)} T_2})})}\\
&\qquad \from
      \ecdSubred
          {(Dev_2 : \ecdEval {(\lfApp {M^\prime} N)} V)} {} {}\\
&\qquad\qquad\qquad
          {(\lfAppTri {Dty_3} {N}
            {(Dty_2 : \ecdOf N {T_1})})}\\
&\qquad\qquad\qquad
          {(Dty : \ecdOf V T)}.\\
&\kwd{sr-abs} : \ecdSubred {\kevabs} {Dty} {Dty}.
\end{smallalign}
\nocaptionrule \caption{\LF signature for \ksubred}
  \label{fig:subred}
\end{figure}

To facilitate totality checking, the user must provide clauses for
deriving typing judgments of the kind in question.
For \ksubred, these clauses might be the ones shown in
Fig~\ref{fig:subred}.
Observe that these clauses essentially describe a recursive method for
constructing a derivation of the output type given ones for the input
types; the object constants \kwd{sr-app} and \kwd{sr-abs} are used to
encode these constructions.

Totality checking in \Twelf is broken into \emph{mode checking},
\emph{termination checking}, \emph{input coverage checking} and
\emph{output coverage checking}; it can be shown that the successful
completion of each of these checks ensures totality of the
non-deterministic function.
Our interest is in extracting a proof from \Twelf's verification
process.
For this, we only need to know the structure of each of the checks and
we elide a discussion of how they guarantee totality.


\subsection{Mode Checking}\label{sec:mode}


Given a type family
$a: \lfTAbsa {x_1:A_1,\ldots,x_n:A_n} {\ktype}$,
we shall refer to the variables $x_1,\ldots,x_n$ as its
\emph{parameters}.
A mode declaration assigns \emph{polarities} $p_1,\ldots,p_n$ to these
parameters, where $p_i$ is either a positive polarity \inMod{} that
designates $x_i$ to be an \emph{input parameter} or a negative polarity
\outMod{} that designates $x_i$ to be an \emph{output parameter}.
A mode declaration for $a$ is well-defined if for any $i$, $1\leq i
\leq n$, such that $p_i = \inMod{}$, the parameters occurring in $A_i$
have polarity $\inMod{}$.
Thus, the input parameters in a type family with a well-defined mode
never depend on its output parameters; this property is necessary for
assigning a meta-theorem reading to the moded type family.
The binder of a type family with a well-defined mode can always be
rearranged so that it has the form
$a:\lfTAbsa {\G^I} {\lfTAbsa {\G^O} {\ktype}}$,
where $\G^I$ and $\G^O$ contain parameters that are assigned only
positive and negative polarities, respectively.
In the following discussion, we will consider only type families with
well-defined modes whose binders also have this special form.
A type family $a$ with a well-defined mode represents the meta-theorem that given any
ground terms $\bar{r}$ for the parameters in $\G^I$, there exists a
derivation $D:\lfAppTri a {\bar{r}} {\bar{s}}$ that computes ground terms
$\bar{s}$ for the parameters in $\G^O$.
We write this meta-theorem formally as
$\mFall {\G^I}
        {\mExst {\G^O}
          {\mExst {D:\lfAppTri a {\G^I} {\G^O}} {\mTrue}}}$.
The requirement that input parameters must not depend on output
parameters provides us a means for extending polarity assignments for
explicit parameters to cover also the implicit parameters.
For example, consider the mode declaration $\kmode\; \ecdSubred
{\inMod{D_1}} {\inMod{D_2}} {\outMod{D_3}}$ that designates the first
two explicit parameters for $\ksubred$ as input the last as output.
The elaborated kind of \ksubred is:
\begin{smallalign}
  \lfTAbs E {\ktm} {
    \lfTAbs V {\ktm} {
      \lfTAbs T {\kty} {
        \ecdEval E V \to \ecdOf E T \to \ecdOf V T \to \ktype
        }}}
\end{smallalign}
Since the implicit parameters $E$, $V$, and $T$ occur in the types of
the first two explicit parameters, well-definedness of mode dictates that they
have the polarity $\inMod{}$.
Correspondingly, the meta-theorem represented by \ksubred is
\begin{smallalign}
\mFall {E:\ktm, V:\ktm, T:\kty, D_1:\ecdEval E V, D_2:\ecdOf E T}\\
\qquad {\mExst {D_3:\ecdOf V T} {\mExst {D:\lfAppSeven {\ksubred} E V
      T {D_1} {D_2} {D_3}} {\mTrue}}}.
\end{smallalign}

Given a type family $a$ with a well-defined mode, mode checking
verifies that the clauses for $a$ that represent a proof for the
relevant meta-theorem respect the moding, \ie, that, given ground
terms for the input parameters, if backchaining on the clause
succeeds, then it will result in ground terms being produced for the
output parameters.
To formalize mode checking, we need the definitions of groundedness
with respect to a context and of input and output
consistency \cite{rohwedder96esop}.
\begin{definition}
An \LF term $M$ is ground with respect to an \LF context $\G$ if all
variables in the canonical form of $M$ are bound in $\G$.
\end{definition}
\noindent Observe that ground terms are a special instance of this
definition, \ie, they are terms that are ground with respect to the
empty context.
\begin{definition}
Let $a : \lfTAbsa {x_1:A_1,\ldots,x_n:A_n} {\ktype}$ be a type family
and let $p_1,\ldots,p_n$ be a well-defined mode for $a$.
We say that
$\lfAppFour a {M_1} {\ldots} {M_n}$
is input consistent relative to $\G$ if for all $i$ such that $1 \leq
i \leq n$ and $p_i = \inMod{}$ it is the case that $M_i$ is ground
with respect to
$\G$.
Similarly,
$\lfAppFour a {M_1} {\ldots} {M_n}$
is output consistent relative to $\G$ if for all $i$ such that $1 \leq
i \leq n$ and $p_i = \outMod{}$ it is the case that $M_i$ is ground
with respect to $\G$.

\end{definition}

A term that $a$ is applied to is called an \emph{input argument} or an
\emph{output argument}, depending on whether it corresponds to an input
parameter or an output parameter.
Variables occurring in input arguments (output arguments) are called
\emph{input variables} (resp. \emph{output variables}).
\begin{definition}\label{def:mode_csst}
Let
$a : \lfTAbsa {\G^I} {\lfTAbsa {\G^O} {\ktype}}$
be a type family with a well-defined mode, where
$\G^I$ is $x_1:A_1^\prime,\ldots,x_k:A_k^\prime$
and $\G^O$ is $x_{k+1}:A_{k+1}^\prime,..,x_n:A_n^\prime$.
Let
$c : A \from A_1 \from \cdots \from A_m$
be a clause for $a$ where
$A = \lfAppFour a {M_1} {\ldots} {M_n}$
and
$A_i = \lfAppFour a {M_{i_1}} {\ldots} {M_{i_n}}$
for $1 \leq i \leq m$.
Let $\G_0$ be the context containing only the input variables that
have a strict occurrence in $A$ and for $1 \leq i \leq m$ let
$\G_i$ be the context $\G_{i-1}, \G_{i}^\prime$, where $\G_i^\prime$
contains only the output variables that have a strict occurrence
in $A_i$.
The clause $c$ is mode consistent if, for $1 \leq i \leq m$, $A_i$
is input consistent relative to  $\G_{i-1}$ and $A$ is output
consistent relative to $\G_m$.
A type family $a$ is well-moded if it has a well-defined mode
and every clause for $a$ is mode consistent.
An \LF signature is well-moded if every type family in it is
well-moded.
\end{definition}
\noindent The restriction to only strict occurrences in the above
definition is based on the fact that the instantiations of
only variables that have such occurrences are guaranteed to be ground
when matching with a ground term.

As an example of the application of these definitions, it is easy to
see that the mode provided for \ksubred is well-defined and
that \kwd{sr-app} and \kwd{sr-abs} are mode consistent.
Thus the type family \ksubred and the \LF signature in
Fig. \ref{fig:subred} are well-moded.
The definition of mode consistency formalizes what is determined by
the mode checking algorithm described in the \Twelf manual
\cite{pfenning02guide}.

\subsection{Termination Checking}\label{sec:terminate}

Termination checking verifies that, given a well-moded \LF signature
and a goal whose input arguments are ground, backchaining on any clause
will result in a finite computation.
Termination is checked using a \emph{termination ordering}.
The fundamental termination ordering used in Twelf is the
\emph{subterm ordering}: $M \subtermOrder N$ if $M$ is a
strict subterm of $N$.
%
%
For example, the totality declaration
$\ktotal\; D\; (\ecdSubred D {\_} {\_})$
tells \Twelf to verify termination of \ksubred using the subterm
ordering on its first argument that has type $\ecdEval E V$.
%
Other orderings such as \emph{lexicographical ordering} and
\emph{simultaneous ordering} are also supported \cite{pfenning02guide}.
Termination checking assumes that the input arguments involved in
termination ordering are ground, a condition that must hold if the
signature is well-moded and mode consistent and the input arguments of
the original goal are all ground.
For every clause, it checks that the input arguments of premises involved
in termination ordering are smaller than corresponding inputs in the
clause head.
For instance, the two premises in \kwd{sr-app} have $Dev_1$ and
$Dev_2$, respectively, as their first (explicit) argument.
These are strict subterms of the input $\ecdEvApp {De_1} {De_2}$ in the
head of \kwd{sr-app}.

The following theorem is proved in \cite{rohwedder96esop}.
\begin{theorem}
Given a well-moded and termination-checked \LF signature $\Si$,
every execution path for a well-typed and input consistent goal $A$
will have only finitely many steps.
\end{theorem}
%
%
By the theorem if a type family $a$ passes mode checking and
termination checking, then any call to $a$ with ground inputs will
terminate.
This does not, however, guarantee that $a$ can be interpreted as a
total function from its ground inputs to its ground outputs: for
some ground inputs the execution might not be able to make progress.
For this stronger guarantee, it must pass the input and output
coverage checking.
%

\subsection{Input Coverage Checking}\label{sec:gencover}
Coverage checking is the general problem of deciding whether any
closed term is an instance of at least one of a given set of
patterns \cite{schurmann03tphols}.
For example, in functional programming languages such as \ML, coverage
checking is used to decide if a function definition is exhaustive over
a given data type.
The complexity of coverage checking depends on the underlying term
algebra.
In \ML which only involves simple types and prefix polymorphism, the
coverage checking is simple.
Coverage checking is significantly more complex in \LF since it
contains dependent types and we also have to consider patterns with
variables of higher-order type.

A coverage goal and a coverage pattern in \LF are both given by a term
and a context with respect to which it is ground.
\begin{definition}
A coverage goal or pattern is a valid \LF typing judgment
$\lfSeq {\G} {\Si} U V$,
where $U$ is either an object or type and $V$ is either a type or
kind.
\end{definition}
A coverage goal or pattern represents the collection of closed terms
obtained by instantiating the variables in the context.
Given a set of coverage patterns and a coverage goal, the task, as
already noted, is to determine if every instance of the goal is an
instance of one of the patterns.
One possibility is that the goal is \emph{immediately covered} by one
of the patterns in the set.
\begin{definition}
A coverage goal
$\lfSeq {\G} {\Si} U V$
is immediately covered by a pattern
$\lfSeq {\G^\prime} {\Si} {U^\prime} {V^\prime}$
if there exists a substitution
$\lfSeq {\G} {\Si} {\sigma} {\G^\prime}$
such that
$\lfSeq {\G} {\Si} {U^\prime[\sigma]} {V^\prime[\sigma]}$
and
$U \equiv U^\prime[\sigma], V \equiv V^\prime[\sigma]$.
A coverage goal is immediately covered by a finite set of patterns if
it is immediately covered by one of the patterns in the set.
\end{definition}
In the more general case, different instances of the coverage goal may
be covered by different patterns.
\begin{definition}\label{def:cover}
$\lfSeq {\G} {\Si} U V$
is covered by a set of patterns $\cal P$
if, for every (ground) substitution $\sigma$ such that
$\lfSeq {\emptyCol} {\Si} {\si} {\G}$,
it is the case that
$\lfSeq {\emptyCol} {\Si} {U[\si]} {V[\si]}$
is immediately covered by ${\cal P}$.
\end{definition}

If it can be shown that the given coverage goal is immediately covered
by the given set of patterns, then the task of coverage checking is
obviously done.
If the goal is not immediately covered by the patterns, then we
consider applying a \emph{splitting} operation to the goal.
Splitting uses knowledge of the signature to generate a set of
subgoals whose simultaneous coverage implies coverage of the original
goal.
We consider here only a restricted form of the operation defined in
\cite{schurmann03tphols} that disallows splitting on variables of
function type and hence the use of context definitions (regular
worlds) in Twelf proofs.
\begin{definition}
Let
$\lfSeq {\G} {\Si} U V$
be a coverage goal and let $\G$ be
$\G_1, x:A_x, \G_2$
where $A_x$ is an atomic type.
Suppose that for every constant
$c:\lfTAbsa {\G_c} {A_c}$
in $\Si$ it is the case that either the unification problem $(A_x =
A_c, x = \lfApp c {\G_c})$ does not have a solution or it has an mgu.
Splitting is then applicable in this case and it generates the
following set of subgoals:
\begin{align*}
& \{
  \lfSeq {\G^\prime, \G_2[\sigma]} {\Si} {U[\sigma]} {V[\sigma]} \sep\\
& \qquad
  c:\lfTAbsa {\G_c} {A_c}\in \Si, (A_x = A_c, x= c\, \G_c) \mbox{ is
    unifiable}\\
& \qquad \mbox{and } \sigma \mbox{ is its mgu} \mbox{ and }
  \lfSeq {\G^\prime} {\Si} {\sigma} {(\G_1,x:A_x,\G_c)}
  \}.
\end{align*}
\end{definition}
%

Note that splitting is finitary, \ie, its application produces only
finitely many subgoals from a given goal. The following theorem,
proved in \cite{schurmann03tphols}, shows that splitting preserves the
coverage:
\begin{theorem}\label{thm:split_prsv}
Let
$\lfSeq {\G} {\Si} U V$
be a coverage goal and let
$\mathset {\lfSeq {\G_i} {\Si} {U_i} {V_i} \sep 1 \leq i \leq n}$
result from it by splitting.
$\lfSeq {\G} {\Si} U V$
is covered by a set of patterns iff
$\lfSeq {\G_i} {\Si} {U_i} {V_i}$
is covered for i such that $1 \leq i \leq n$.
\end{theorem}
%

%
The general notion of coverage checking described here is specialized
in \Twelf to input coverage checking that verifies
that there is some clause to backchain on for any call to $a$
with ground inputs.
The coverage patterns and goal to be used in the framework of general
coverage checking for this purpose are defined as follows.
\begin{definition}
For a given type family
$a:\lfTAbsa {\G^I} {\lfTAbsa {\G^O} {\ktype}}$
in a well-moded \LF signature $\Si$,  the input coverage
goal is $\G^I \lfseq a\, \G^I : \Pi\G^O.  \ktype$.
Corresponding to a clause $c$ for $a$ whose head is $a\, M_1\, \ldots\, M_n$,
let $\si^I_c$ be the substitution $(M_1/x_1,\ldots,M_k/x_k)$ where $k$
is the length of $\G^I$ and let $\G^I_c$ be the context containing the
variables occurring in $M_1,\ldots,M_k$.
The set of input patterns for $a$ relative to $\Si$ then consists of
the following:
\begin{align*}
\mathset{
  \lfSeq {\G^I_c} {\Si} {\lfApp a {\si^I_c}}
         {\lfTAbsa {\G^O[\si^I_c]} {\ktype}} \sep
  \mbox{$c$ is a clause for $a$ in $\Si$}
}.
\end{align*}
\end{definition}
The input coverage patterns for $a$ are derived from the clauses for $a$
by fixing the input arguments to $a$ to be those for the head of the
respective clause.
Moreover, these arguments in the input coverage goal are set to be
variables.
Thus, if coverage checking as per Definition \ref{def:cover} verifies
that this goal is covered by the corresponding patterns, then any
call to $a$ with ground input arguments and variables for the output
arguments will match the head of at least one clause for $a$, thus
guaranteeing the possibility of backchaining on it.
%

A terminating procedure for (input) coverage checking that is based on a
repeated use of splitting and immediate coverage is described for \LF
in \cite{schurmann03tphols}.
We do not present this procedure here since our interest is in the
extraction of an explicit proof from the results of totality
checking.
In particular, we assume that we are presented at the outset with a
sequence of splitting operations applied to an input coverage goal
that lead to a set of subgoals that pass immediate coverage checking.

\subsection{Output Coverage Checking}\label{sec:outputcover}
The premises in a clause for the type family correspond to
recursive calls.
The output arguments in such recursive calls could potentially limit
the success set.
Since coverage checking based on input coverage assumes success any
time the input arguments in a goal match those in the head of a
clause, it is necessary to verify that the output arguments of
premises do not falsify this assumption.
We describe here a method for ensuring that this is the case; this
presentation is a formalization of what we understand output coverage
checking to be in Twelf \cite{pfenning02guide}.

One way in which the success set may get constrained is if a
output variable appears in an input argument. \emph{Output freshness}
is a criterion designed to avoid such a possibility.
\begin{definition}\label{def:output_freshness}
A clause $c: A \from A_1 \from \cdots \from A_m$ satisfies the output freshness
property iff, for $1 \leq i \leq m$, the sets of output and input
variables of $A_i$ are disjoint.
\end{definition}

We assume output freshness for clauses in what follows. A more
involved requirement is that the form of the output
arguments not limit the coverage of the clause. To formalize this we
first identify output coverage goals and patterns.
\begin{definition}\label{def:out_cover}
Let
$a:\lfTAbsa {\G^I} {\lfTAbsa {\G^O} {\ktype}}$
be a type family in a well-moded \LF signature $\Si$, where
$\G^I$ is $x_1:A_1^\prime,\ldots,x_k:A_k^\prime$ and
$\G^O$ is $x_{k+1}:A_{k+1}^\prime, \ldots, x_n:A_n^\prime$.
Further, let
$c: A \from A_1 \from \cdots \from A_m$
be a clause in $\Si$, where
$A$ is $\lfAppFour a {M_{1}} {\ldots} {M_{n}}$
and, for $1 \leq i \leq m$,
$A_i$ is $\lfAppFour a {M_{i_1}} {\ldots} {M_{i_n}}$.
Finally, for $1 \leq i \leq m$, let
$\si^I_i$ be the substitution $(M_{i_1}/x_1,\ldots,M_{i_k}/x_k)$.
Then the output coverage pattern for $A_i$ is
$\lfSeq {\G^I_i, \G^O_i} {\Si} {A_i} {\ktype}$
where $\G^I_i$ and $\G^O_i$ are contexts formed respectively from the
input and output variables of $A_i$.
The output coverage goal for $A_i$ is
$\lfSeq {\G_i^I;\G^O[\si^I_i]} {\Si} {\lfAppTri a {\si^I_i} {\G^O}} {\ktype}$.
\end{definition}

Intuitively, the output coverage pattern corresponding to a premise
$A_i$ is $A_i$ itself.
The output coverage goal, on the other hand, is obtained by retaining
the inputs to $A_i$ while maximally generalizing its outputs.
To check if an output coverage goal is covered by the single pattern,
we may split on an output variable but a way that ensures no variable
other than the one being split on is instantiated.
%
\begin{definition}\label{def:out_splitting}
Let
$\lfSeq {\G^I ; \G^O} {\Si} A {\ktype}$
be an output coverage goal where $\G^O$ is
$\G_1, x:A_x, \G_2$
for some atomic type $A_x$.
Suppose that for every constant
$c:\lfTAbsa {\G_c} {A_c}$
in $\Si$ it is the case that either
$A_x$ is not an instance of $A_c$
or there is a most general substitution $\sigma$ such that $A_x =
A_c[\sigma]$; such a substitution must obviously not instantiate any
variable in $\G^I,\G^O$.
%
Then, output splitting generates the set of subgoals
\begin{align*}
& \{
  \lfSeq {\G^I; \G_1, \G_c^\prime, \G_2[\sigma]} {\Si} {A[\sigma]} {\ktype} \sep\\
& \qquad
  c:\lfTAbsa {\G_c} {A_c} \in \Si,  A_x \mbox{ is an instance of }
  A_c, \sigma' \mbox{ is a}\\
& \qquad \mbox{most general substitution such that }
  A_x = A_c[\sigma'],\\
& \qquad \sigma = ((c\, \G_c)[\sigma']/x,\sigma') \mbox{ and } \lfSeq {\G_c^\prime} {\Si} {\sigma'} {\G_c)}.
  \}
\end{align*}
\end{definition}
Eventually we want every output coverage subgoal that is produced by
splitting to be equivalent to the sole output coverage pattern for each premise.
Immediate output coverage captures the relevant equivalence
notion.
\begin{definition}\label{def:out_immed_cover}
An output coverage goal
$\lfSeq {\G^I; \G^O} {\Si} {M^\prime} {\ktype}$
is immediately covered by an output coverage pattern
$\lfSeq {\G^\prime} {\Si} {M} {\ktype}$
if
$M \equiv M^\prime[\sigma]$ for a substitution $\sigma$ that only
renames variables and is such that $\lfSeq {\G} {\Si} {\sigma}
{\Gamma^I,\Gamma^O}$.
\end{definition}
An output coverage goal is covered by an output coverage pattern if
every subgoal that is produced from it by some applications of output
splitting is immediately covered by the pattern.
In summary, output coverage checking ensures output freshness for
every clause in the signature and it further checks that the output
coverage goal for every premise of every clause is covered by the
corresponding output coverage pattern for the premise.

\section{Explicit Proofs for Meta-Theorems}\label{sec:m2}

\long\def\mtwobasics{
  \begin{figure}[ht!]
    \centering
    \begin{smallgather}
      \infer[\kfalll]{
        \mSeq
            {\G}
            {\D_1, \mTm {\boxed{x}} {\mFall {\G_1} {F_1}}, \D_2}
            {\mTm {\boxed{\mLet {y=\lfApp x {\si}} P}} {F_2}}
      }{
        \lfSeq {\G} {\Si} {\si} {\G_1}
        &
        \mSeq
            {\G}
            {
              \D_1,
              \mTm {\boxed{x}} {\mFall {\G_1} {F_1}},
              \D_2,
              \mTm {\boxed{y}} {F_1[\sigma]}
            }
            {\mTm {\boxed{P}} {F_2}}
      }
    \end{smallgather}
    \begin{smallgather}
      \infer[\kexistsl] {
        \mSeq
            {\G}
            {\D_1, \mTm {\boxed{x}} {\mExst {\G_1} {\mTrue}}, \D_2}
            {\mTm {\boxed{\mSplit x {\G_1} P}} F}
      }{
        \mSeq
            {\G,\G_1}
            {\D_1, \mTm {\boxed{x}} {\mExst {\G_1} {\mTrue}}, \D_2}
            {\mTm {\boxed{P}} F}
      }
    \end{smallgather}
    \begin{smallgather}
      \infer[\kfallr]{
        \mSeq {\G} {\D} {\mTm {\boxed{\mIAll {\G_1} P}} {\mFall {\G_1} F}}
      }{
        \mSeq {\G,\G_1} {\D} {\mTm {\boxed{P}} F}
      }
      \qquad
      \infer[\kexistsr]{
        \mSeq {\G} {\D}
              {\mTm {\boxed{\mSubst {\sigma}}} {\mExst {\G_1} {\mTrue}}}
      }{
        \lfSeq {\G} {\Si} {\sigma} {\G_1}
      }
    \end{smallgather}
\nocaptionrule  \caption{Quantifier Rules}
  \label{fig:m2basics}
  \end{figure}
}

\long\def\caserules{
  \begin{figure*}[ht!]
    \centering
    \begin{smallgather}
      \infer[\ksigempty] {
        \mcSeq {\G_1, x: A_x, \G_2} {\D} {\emptyCol} {\mTm {\boxed{\mNull}} F}
      }{
      }
      \qquad
      \infer[\ksignonuni, A_x \mbox{ and } A_c \mbox{ do not unify}]{
        \mcSeq
            {\G_1, x:A_x, \G_2}
            {\D}
            {\Si,c: \lfTAbsa {\G_c} {A_c}}
            {\mTm {\boxed{\O}} F}
      }{
        \mcSeq {\G_1, x:A_x, \G_2} {\D} {\Si} {\mTm {\boxed{\O}} F}
      }
    \end{smallgather}
    \begin{smallgather}
      \infer[\ksiguni]{
        \mcSeq
            {\G_1, x:A_x, \G_2}
            {\D}
            {\Si,c:\lfTAbsa {\G_c} {A_c}}
            {\mTm
              {\boxed{
                  \O,
                  (\mItem
                    {\mPattern {\G^\prime} {\G_2[\sigma]} {(\lfApp c {\G_c})[\sigma]}}
                    P)}}
              F}
      }{
        \mSeq {\G^\prime, \G_2[\sigma]} {\D[\sigma]} {\mTm {\boxed{P}} {F[\sigma]}}
        &
        \mcSeq {\G_1,x:A_x,\G_2} {\D} {\Si} {\mTm {\boxed{\O}} F}
      }
      \\
      \lfSeq {\G^\prime} {} {\sigma} {(\G_1, x:A_x, \G_c)} \text{ where }
      \sigma = mgu(A_x = A_c, x= c\, \G_c)
    \end{smallgather}
\nocaptionrule    \caption{Rules for
      $\mcSeq {\G_1, x:A_x, \G_2} {\D} {\Si} {\mTm {\boxed{\O}} F}$}
    \label{fig:case}
  \end{figure*}
}

Meta-theorems about formal systems can also be stated and proved
in a logic.
The logic \M2 is designed for doing this based on the \LF encodings
of such systems \cite{schurmann98cade}.
\M2 is a constructive logic formally presented via a sequent calculus.
Formulas in \M2 have the form $\mFall {\G_1} {\mExst {\G_2} {\mTrue}}$,
where $\G_1$ and $\G_2$ are valid \LF contexts.
The universal (existential) quantification is omitted when $\G_1$
($\G_2$) is empty.
As we can see, the formulas in \M2 are limited to the $\Pi_0^1$ form,
where all existential quantifiers follow the universal ones.
Although this may not seem very expressive, we have seen in
Section~\ref{sec:total} that every theorem proved through totality
checking has this form.

The judgments in the sequent calculus are of the form
$\mSeq {\G} {\D} {\mTm P F}$,
where $F$ is a formula, $P$ is a \emph{proof term}, $\D$ is a set of
\emph{assumptions} and $\G$ is a valid \LF context containing all free
variables occurring in $P$ and $F$.
Proof terms and assumptions are defined as follows:
\begin{align*}
  \text{Proof Terms:} & \qquad
  P ::=
  \mLet {y = \lfApp x {\si}} P \sep
  \mIAll {\G} P \sep\\
  & \qquad\qquad
  \mSplit x {\G} P \sep
  \mSubst {\sigma}\\
  \text{Assumptions:} & \qquad
  \D ::= \emptyCol \sep \D, \mTm P F
\end{align*}
A meta-theorem represented by the formula
$\mFall {\G_1} {\mExst {\G_2} {\mTrue}}$,
is proved by  deriving the judgment
$\mSeq {\emptyCol} {\emptyCol} {\mTm P {\mFall {\G_1} {\mExst {\G_2} {\mTrue}}}}$.
The proof term $P$ is obtained as an output of the
derivation.
The resulting $P$ represents a total function, as shown by the
following theorem that is proved in \cite{schurmann00phd}.
\begin{theorem}
If
$\mSeq {\emptyCol} {\emptyCol} {\mTm P {\mFall {\G_1} {\mExst {\G_2} {\mTrue}}}}$
is derivable for some $P$, then for every closed substitution
$\lfSeq {\emptyCol} {\Si} {\sigma_1} {\G_1}$
there exists a substitution
$\lfSeq {\emptyCol} {\Si} {\sigma_2} {\G_2[\sigma_1]}$.
\end{theorem}
As an example of this theorem, recall from Section~\ref{sec:total} the
formula stating the subject reduction theorem:
\begin{smallalign}
\mFall {E:\ktm, V:\ktm, T:\kty, D_1:\ecdEval E V, D_2:\ecdOf E T}\\
\qquad {\mExst {D_3:\ecdOf V T} {\mExst {D:\lfAppSeven {\ksubred} E V
      T {D_1} {D_2} {D_3}} {\mTrue}}}.
\end{smallalign}
%
If we can get a proof term for this formula, then we can conclude that
for any closed terms
$D_1:\ecdEval E V$
and
$D_2:\ecdOf E T$,
there exists a term
$D_3:\ecdOf V T$.
Given the adequacy of the \LF encoding, we can conclude the subject
reduction theorem holds for the actual system.

The subsections that follow present the derivation rules for \M2.
%

\subsection{The Quantifier Rules}

\mtwobasics

The quantifier rules for \M2 that are presented in
Fig. \ref{fig:m2basics} are the most basic ones for the logic.
If we ignore the proof terms in boxes, we can see that these rules are
similar to the conventional ones for an intuitionistic logic.
The \kexistsl and \kfallr rules introduce fresh eigenvariables to the
context.
The \kfalll rule instantiates an assumption with a witnessing substitution
to get a new assumption.
\kexistsr finds a witnessing substitution and finishes the proof.
Note that weakening is implicit in $\kexistsr$.

\subsection{Recursion}

The recursion rule is
\begin{align*}
  \infer[\krecur]{
    \mSeq {\G} {\D} {\mTm{\boxed{\mRecur x F P}} F}
  }{
    \mSeq {\G} {\D, \mTm {\boxed{x}} F} {\mTm {\boxed{P}} F}
  }
\end{align*}
with the proviso that $\mRecur x F P$ must terminate in $x$.
This rule adds the goal formula as an inductive
hypothesis to the set of assumptions.
For a proof based on this rule to be valid, the proof term must
represent a terminating computation as the side condition guarantees.
This condition is presented formally in Definition 7.8 of
\cite{schurmann00phd} using the termination ordering on
\LF terms that was discussed in Section~\ref{sec:terminate}.

\subsection{Case Analysis}\label{sec:m2case}

The case analysis rule considers all the possible top-level structures
for a ground term instantiating an eigenvariable $x$ of atomic type
$A_x$ in the context $\G$ in a judgment of the form $\G; \D \m2seq P \in
F$.
Since $A_x$ is atomic, the only cases to consider are those where the
head of the term is a constant from the \LF signature.
%
%

To state the rule, we need to extend the definition of proof terms
with a case construct:
\begin{align*}
  \text{Patterns:} \qquad &
  R ::= \mPattern {\G^\prime} {\G^{\prime\prime}} M\\
  \text{Cases:}   \qquad  &
  \O ::= \emptyCol \sep \O, \mItem R P\\
  \text{Proof Terms:} \qquad  &
  P ::= ... \sep \mCase x {\O}
\end{align*}
The case rule is
\begin{gather*}
  \infer[\kcase]{
    \mSeq {\G} {\D} {\mTm {\boxed{\mCase x {\O}}} F}
  }{
    \mcSeq {\G_1, x:A_x, \G_2} {\D} {\Si} {\mTm {\boxed{\O}} F}
  }
\end{gather*}
where $\G_1, x: A_x, \G_2$ is a valid permutation of $\G$ and $A_x$
depends on all the variables in $\G_1$.
The new judgment
$\mcSeq {\G_1, x:A_x, \G_2} {\D} {\Si} {\mTm {\O} F}$
is derived by considering all constants in $\Si$ whose target type
unifies with $A_x$.
The derivation rules for this judgment are shown in Fig. \ref{fig:case}.
The \ksiguni rule produces a new premise (and a new case in the proof
term) for every constant $c$ in $\Si$ that has a type $\Pi
\G_c. A_c$ such that $A_c$ unifies with $A_x$.
The situation where $A_x$ and $A_c$ do not unify is dealt with by the
\ksignonuni rule.

\caserules

\subsection{Interpretation of Proof Rules}
At an intuitive level, the proof rules of \M2 function as follows:
\begin{itemize}
\item \krecur introduces an inductive hypothesis as an assumption,
  usually as the first step, reading proofs upwards;
\item \kfalll followed by \kexistsl corresponds to an application of
  the inductive hypothesis;
\item \kfallr introduces hypotheses into the context;
\item \kcase corresponds to case analysis on a hypothesis;
\item \kexistsr finishes a proof branch by constructing witness from
  the context.
\end{itemize}

\newcommand{\exmTypeFamily}
           {a:\lfTAbsa {\G^I} {\lfTAbsa {\G^O} {\ktype}}}
\newcommand{\exmFormula}
           {\mFall {\G^I}
             {\mExst {\G^O}
               {\mExst {D: \lfAppTri a {\G^I} {\G^O}} {\top}}}}
\newcommand{\exmEndSequent}
           {\mSeq {\emptyCol} {\emptyCol} {\mTm {P} {\exmFormula}}}
\newcommand{\IH}
           {\mFall {\G^I}
              {\mExst {\G^O}
                {\mExst {D:\lfAppTri a {\G^I} {\G^O}} {\top}}}}
\newcommand{\siI}{\si^I_c}
\newcommand{\siO}{\si^O_c}
\newcommand{\fBranch}{\mExst {\G^O[\si^I_c]}
  {\mExst {D:\lfAppTri a {\si^I_c} {\G^O}} {\top}}}
\newcommand{\Gc}[1]{\G_{c_{#1}}}
\newcommand{\Sc}[1]{S_{c_{#1}}}

\section{Explicit Proofs from Totality Checking}\label{sec:genproofs}
Totality checking is a proof verification tool that asserts if the
user-defined \LF signatures represent proofs for meta-theorems.
This section describes an algorithm for generating \M2 proofs from
totality checking that does not use contexts and lemmas and the
proof for its correctness.
The algorithm consists of two steps.
Firstly, a well-moded type family, which can always be stated as
$\exmTypeFamily$ where $\G^I$ contains input parameters and $\G^O$
contains output parameters, is translated into an \M2 formula
$\exmFormula$.
Secondly, an \M2 proof tree for the sequent $\exmEndSequent$ is
generated by applying \M2 proof rules translated from totality
checking until a complete proof tree is built.
The validity of those applications is guaranteed by totality checking,
as we will show in the correctness proof for the algorithm.
In the following subsections, we first give an overview of the proof
generation process.
Then we discuss the individual steps for constructing the proof tree.
Finally, we prove the correctness of the proof generation algorithm.

\subsection{Overview of the Proof Generation}
Totality checking verifies that an \LF signature provided by user
represents a total function that computes ground outputs from ground
inputs.
This interpretation of proofs is translated into an \M2 proof tree as
shown in Fig. \ref{fig:overview}.
\begin{figure}[h!]
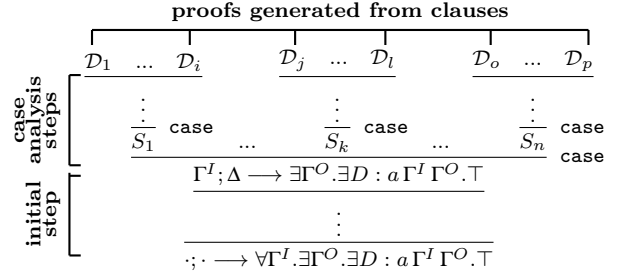

  \center
\begin{pgfpicture}{0cm}{0cm}{8.5cm}{3.6cm}
\pgfputat
    {\pgfxy(1.2,0.1)}
    {\pgfbox[left,bottom]{{\small
        $\infer[] {
          \mSeq {\emptyCol} {\emptyCol} {\exmFormula}
        }{
          \infer[]{\vdots}{
            \infer[\kcase]{
              \mSeq {\G^I}
                  {\D}
                  {\mExst {\G^O} {\mExst {D:\lfAppTri a {\G^I} {\G^O}} {\top}}}}{
              \infer[\kcase]{S_1}{
                \infer[]{\vdots}{
                  \drv_1 & ... & \drv_i
                }
              }
              &
              ...
              &
              \infer[\kcase]{S_k}{
                \infer[]{\vdots}{
                  \drv_j & ... & \drv_l
                }
              }
              &
              ...
              &
              \infer[\kcase]{S_n}{
                \infer[]{\vdots}{
                  \drv_o & ... & \drv_p
                }
              }
            }
          }
        }$
    }}}

    {\pgfsetlinewidth{0.8pt}}

  {\pgfline{\pgfxy(1.3,3.1)}{\pgfxy(1.3,3.3)}}
  {\pgfline{\pgfxy(2.6,3.1)}{\pgfxy(2.6,3.3)}}
  {\pgfline{\pgfxy(4,3.1)}{\pgfxy(4,3.3)}}
  {\pgfline{\pgfxy(5.2,3.1)}{\pgfxy(5.2,3.3)}}
  {\pgfline{\pgfxy(6.6,3.1)}{\pgfxy(6.6,3.3)}}
  {\pgfline{\pgfxy(1.3,3.3)}{\pgfxy(7.9,3.3)}}
  {\pgfline{\pgfxy(7.9,3.3)}{\pgfxy(7.9,3.1)}}
  {\pgfputat{\pgflabel{.5}{\pgfxy(1.3,3.3)}{\pgfxy(7.9,3.3)}{5pt}}{\pgfbox[center,base]{{\small \textbf{proofs generated from clauses}}}}}

  {\pgfline{\pgfxy(1,2.7)}{\pgfxy(1.15,2.7)}}
  {\pgfline{\pgfxy(1,2.7)}{\pgfxy(1,1.5)}}
  {\pgfline{\pgfxy(1,1.5)}{\pgfxy(1.15,1.5)}}
  {\pgfputlabelrotated{.5}{\pgfxy(1,1.5)}{\pgfxy(1,2.7)}{5pt}{\pgfbox[center,base]{{\small \textbf{steps}}}}}
  {\pgfputlabelrotated{.5}{\pgfxy(0.79,1.5)}{\pgfxy(0.79,2.7)}{5pt}{\pgfbox[center,base]{{\small \textbf{analysis}}}}}
  {\pgfputlabelrotated{.5}{\pgfxy(0.58,1.5)}{\pgfxy(0.58,2.7)}{5pt}{\pgfbox[center,base]{{\small \textbf{case}}}}}

  {\pgfline{\pgfxy(1,1.36)}{\pgfxy(1.15,1.36)}}
  {\pgfline{\pgfxy(1,1.4)}{\pgfxy(1,0.3)}}
  {\pgfline{\pgfxy(1,0.3)}{\pgfxy(1.15,0.3)}}
  {\pgfputlabelrotated{.5}{\pgfxy(1,0.3)}{\pgfxy(1,1.4)}{5pt}{\pgfbox[center,base]{{\small \textbf{step}}}}}
  {\pgfputlabelrotated{.5}{\pgfxy(0.79,0.3)}{\pgfxy(0.79,1.4)}{5pt}{\pgfbox[center,base]{{\small \textbf{initial}}}}}

\end{pgfpicture}
\nocaptionrule  \caption{Overview of the Proof Generation}
  \label{fig:overview}
\end{figure}
The proof tree consists of several segments, each consisting of
applications of proof rules translated from different components in
totality checking.
It is constructed as follows:
first there is an initial step that introduces an inductive hypothesis
and other hypotheses.
Then follows the case analysis steps that translate input coverage
checking into applications of \kcase.
Then every clause for $a$ in the \LF signature is translated into a proof
branch.
For each frontier sequent resulting from case analysis steps,
there exists a proof branch that, when instantiated, is the proof for
that sequent.
Plugging in the proof for those frontier sequents finishes
the proof.

For simplicity we omit proof terms in \M2 sequents in the following
discussions.
The translation can be extended to generate \M2 proofs with proof
terms in a straightforward manner.
The proof term for the end sequent in a translated proof represents a
total function, whose termination property is ensured by termination
checking.

\subsection{Initial Step}
Starting with the initial sequent $\mSeq {\emptyCol} {\emptyCol}
{\exmFormula}$, we apply the \krecur rule to introduce the goal
formula as an inductive hypothesis.
Then we apply the \kfallr rule to introduce hypotheses in $\G^I$ into
the \LF context.
\begin{figure}[h!]
\center
\begin{smallgather}
  \infer[\krecur]{
    \mSeq {\emptyCol} {\emptyCol}
          {\IH}
  }{
    \infer[\kfallr]{
      \mSeq {\emptyCol}
            {\IH}
            {\IH}
    }{
      \mSeq {\G^I}
            {\IH}
            {\mExst {\G^O} {\mExst {D:\lfAppTri a {\G^I} {\G^O}} {\top}}}
    }
  }
\end{smallgather}
\nocaptionrule \caption{Introducing Hypotheses}
\label{fig:init}
\end{figure}

\subsection{Case Analysis Steps}
In this phase, we translate input coverage checking into applications
of the \kcase rule.
We assume that input coverage checking succeeds and returns a
sequence of splitting operations that leads to subgoals covered by
input patterns.
We show that every splitting operation can be translated into an
application of \kcase.
In Section \ref{sec:gencover}, we described
that splitting on the variable $x$ in an input coverage goal
$\lfSeq {\G^I_1, x:A_x, \G^I_2} {\Si} A {\lfTAbsa {\G^O} {\ktype}}$
produces a subgoal
$\lfSeq {\G^\prime, \G^I_2[\sigma]} {\Si}
  {A[\sigma]} {\lfTAbsa {\G^O[\sigma]} {\ktype}}$
for every constant $c:\lfTAbsa {\G_c} {A_c}$ such that $A_c$ and $A_x$ unifies
and
$\lfSeq {\G^\prime} {\Si} {\sigma = mgu(A_x = A_c, x= \lfApp c {\G_c})}
{(\G^I_1,x:A_x,\G_c)}$.
From the definitions of \kcase in Section \ref{sec:m2case}, we
can see that splitting on a variable $x:A$ corresponds exactly to case
analysis of $x$.
To formally prove the equivalence, we define a relation between input
coverage goals and \M2 sequents:
\begin{definition}
Let $\rsc$ be a binary relation between input coverage goals and \M2
sequents.
$G \rsc S$ holds if and only if $G$ is an input coverage goal
$\lfSeq {\G^I} {\Si} A {\lfTAbsa {\G^O} {\ktype}}$
and $S$ is an \M2 sequent
$\mSeq {\G^I} {\D} {\mExst {\G^O} {\mExst {D : \lfApp A {\G^O}} {\top}}}$
where $\D = \IH$.
\end{definition}
We overload the $\rsc$ relation with finite sets of coverage goals and
\M2 sequents:
\begin{definition}
Given a finite set of input coverage goals
$\GSet = \mathset{G_i \sep 1 \leq i \leq n}$
and \M2 sequents
$\SSet = \mathset{S_i \sep 1 \leq i \leq n}$,
$\GSet \rsc \SSet$
if and only if $G_i \rsc S_i$ for $1\leq i \leq n$.
\end{definition}
Then we have the following lemma:
\begin{lemma}\label{lem:rsc}
Given $G \rsc S$ where $G$ is
$\lfSeq {\G^I_1, x:A_x, \G^I_2} {\Si} A {\lfTAbsa {\G^O} {\ktype}}$
and $S$ is
$\mSeq {\G^I_1, x:A_x, \G^I_2} {\D}
  {\mExst {\G^O}
    {\mExst {D: \lfApp A {\G^O}} {\top}}}$,
if $\GSet$ is the set of subgoals resulting from splitting on $x$ in $G$ and
$\SSet$ is the set of frontier sequents resulting from applying \kcase to $x$
in $S$,
then $\GSet \rsc \SSet$.
\end{lemma}
\begin{proof}
Let $G_i$ be a subgoal
$\lfSeq {\G^\prime, \G^I_2[\si]} {\Si} {A[\si]}
  {\lfTAbsa {\G^O[\sigma]} {\ktype}}$
resulting from splitting for some constant $c:\lfTAbsa {\G_c} {A_c}$
such that $A_c$ and $A_x$ unifies and
$\lfSeq {\G^\prime} {\Si} {\sigma = mgu(A_x = A_c, x= \lfApp c {\G_c})}
  {(\G^I_1,x:A_x,\G_c)}$.
By the definition of case analysis, there must be a sequent
$S_i = \mSeq {\G^\prime, \G^I_2[\sigma]} {\D}
  {\mExst {\G^O[\sigma]}
    {\mExst {D: \lfApp {A[\si]} {\G^O[\si]}} {\top}}}$
resulting from applying \ksiguni.
Thus we have $G_i \rsc S_i$ by the definition of $\rsc$.
Since goals in $\GSet$ one-to-one correspond to sequents in $\SSet$,
we have $\GSet \rsc \SSet$.
\end{proof}
The following lemma leads to an algorithm for translating splittings
to applications of \kcase:
\begin{lemma}\label{lem:rscproof}
Given a set of input coverage goals $\GSet$ and a partial \M2 proof
with frontier sequents $\SSet$ such that $\GSet \rsc \SSet$, if
for some $G_k \in \GSet$ such that
$G_k = \lfSeq {\G^I_1, x:A_x, \G^I_2} {\Si} A {\lfTAbsa {\G^O} {\ktype}}$,
splitting on $x$ results in a new set of coverage goals $\GSet'$, then
applying the \kcase rule on $x$ in $S_k$ where $G_k \rsc S_k$ results
in a partial \M2 proof with frontier sequents $\SSet'$ such that
$\GSet' \rsc \SSet'$.
\end{lemma}
\begin{proof}
The coverage goals $G_i \in \GSet$ for $i \neq k$ and frontier
sequents $S_i \in \SSet$ for $i \neq k$ are still present and related
by $\rsc$ after splitting and applying \kcase.
The new coverage goals generated by splitting are related to the new
frontier sequents generated by applying \kcase by Lemma \ref{lem:rsc}.
\end{proof}
Starting with the partial \M2 proof resulting from the initial step,
we translate input coverage checking into applications of \kcase.
We maintain a set of input coverage goals $\GSet$ and a partial proof
with frontier sequents $\SSet$ such that $\GSet \rsc \SSet$ for the
translation.
Initially, $\GSet$ is a singleton containing
$\lfSeq {\G^I} {\Si} {\lfApp a {\G^I}} {\lfTAbsa {\G^O} {\top}}$
and $\SSet$ is a singleton containing
$\mSeq {\G^I} {\IH} {\mExst {\G^O} {\mExst {D: \lfAppTri a {\G^I}
      {\G^O}} {\top}}}$
that is the frontier sequent of the proof in Figure
\ref{fig:init}.
By Lemma \ref{lem:rscproof}, for every splitting operation on some
$G_i \in \GSet$, we are able to apply the \kcase rule to the
corresponding frontier sequent $S_i \in \SSet$ where $G_i \rsc S_i$ to get
$\GSet'$ and $\SSet'$ such that $\GSet' \rsc \SSet'$.
In the end, we get a set of input coverage goals $\GSet'$ that are
immediately covered and a partial proof tree with frontier sequents
$\SSet'$ such that $\GSet' \rsc \SSet'$.
%


\subsection{Generating Proof Trees from Clauses}\label{sec:gen_proof_from_clauses}
Let
$\GSet' = \mathset {G_i \sep 1\leq i \leq n^\prime}$
and
$\SSet' = \mathset{S_i \sep 1\leq i \leq n^\prime}$, where $\GSet' \rsc
\SSet'$, be the input coverage goals and sequents resulting from case
analysis steps.
To finish the \M2 proof, we need to derive sequents in $\SSet'$.
For this we define the instantiation of an \M2 sequent:
\begin{definition}\label{def:seq_inst}
Let $S$ be an \M2 sequent $\mSeq {\G} {\D} F$ and $\si$ be a
substitution such that $\lfSeq {\G^\prime} {\Si} {\si} {\G}$.
The instantiation of $S$ under $\si$ is
$\mSeq {\G^\prime} {\D[\si]} {F[\si]}$.
\end{definition}
The following lemma shows that immediate coverage in input coverage
checking can be reflected into instantiation of \M2 sequents.
\begin{lemma}\label{lem:in_immed_cover_inst}
Given an input coverage goal $G^\prime$ and an input coverage pattern
$G$ such that $G^\prime$ is immediately covered by $G$ under a
substitution $\si$, if $S$ and $S'$ are \M2 sequents for which $G \rsc
S$ and $G^\prime \rsc S^\prime$ hold, then $S^\prime$ is the instantiation
of $S$ under $\si$.
\end{lemma}
\begin{proof}
Straightforward from definitions of immediate coverage, instantiation and
$\rsc$ relation.
\end{proof}
Let $G_i \in \GSet$ and $S_i \in \SSet$ where $G_i \rsc S_i$ and $G
\rsc S$ where $G$ is a pattern derived from a clause $c$ that
immediately covers $G_i$.
We are going to translate $c$ into an \M2 proof with the end sequent $S$.
By Lemma \ref{lem:in_immed_cover_inst}, the instantiation of the proof
for $S$ under the substitution for the immediate coverage of $G_i$ by
$G$ is a proof for $S_i$.
The instantiation of a proof generated from a clause is also a valid
proof, which we will prove after describing the translation of
clauses.

\subsubsection{Translation of Clauses into Proofs}\label{sec:trans_proof_from_clauses}
We first describe the translation of clauses into \M2 proofs.
Let $a : \lfTAbsa {\G^I} {\lfTAbsa {\G^O} {\ktype}}$
be a type family where $\G^I = x_1:A_1, ..., x_k:A_k$ and $\G^O =
x_{k+1}:A_{k+1},...,x_n:A_n$,
$c: a\, M_1\, ...\, M_n \from a\, M_{1_1}\, ...\, M_{1_n} \from
... \from a\, M_{m_1}\, ...\, M_{m_n}$
be a clause for $a$,
$\siI = (M_1/x_1,...,M_k/x_k)$
be the substitution containing input arguments to the head and
$\G^I_c$ be the context containing variables in $\siI$.
The input coverage pattern $G$ derived from $c$ is
$\lfSeq {\G^I_c} {\Si} {\lfApp a {\siI}} {\lfTAbsa {\G^O[\siI]} {\ktype}}$.
Let $S$ be the sequent
\begin{gather*}
\mSeq {\G^I_c} {\D} {\fBranch}
\end{gather*}
such that $G \rsc S$, where $\D$ contains the inductive hypothesis $\IH$.
%
%
The \M2 proof for $S$ is translated from the operational
interpretation of $c$ as a function that computes ground outputs from
ground inputs.
Starting with the input arguments to the head, this function
recursively calls $a$.
Each recursive call to $a$ is represented by a premise
$\lfAppFour a {M_{i_1}} {...} {M_{i_n}}$
for some $i$ such that $1 \leq i \leq m$.
After the recursive calls, it constructs outputs from initial inputs
and the results of recursive calls.
We describe the translation from $c$ to a proof for $S$ as a recursive
procedure.
In the base step, we translate the construction of outputs into
an application of the \kexistsr rule.
In each recursive step, we translate a recursive call to $a$ into
applications of \kfalll, \kexistsl and \kcase rules and recursively
generate the sub-proof.

Before describing the translation steps, we present some notations
that will be used in our discussion and an alternative definition of
output coverage checking in order to simplify the discussion.
\begin{definition}\label{def:const_trans}
Let
$a : \lfTAbsa {\G^I} {\lfTAbsa {\G^O} {\ktype}}$
be a type family with well-defined mode, where
$\G^I = x_1:A_1^\prime,...,x_k:A_k^\prime$
and
$\G^O = x_{k+1}:A_{k+1}^\prime,..,x_n:A_n^\prime$.
Let
$c : A \from A_1 \from ... \from A_m$
be a clause for $a$ where
$A = \lfAppFour a {M_1} {...} {M_n}$
and
$A_i = \lfAppFour a {M_{i_1}} {...} {M_{i_n}}$
for $1 \leq i \leq m$.
Let $\siI = (M_1/x_1,...,M_k/x_k)$ be the substitution containing input
arguments to the head of $c$.
A set of contexts $\mathset{\Gc{i} \sep 0 \leq i \leq m}$ is
recursively defined as follows: $\Gc{0}$ contains input variables in $A$
and $\Gc{i} = \Gc{i-1}, \G_i^\prime$ for $1\leq i\leq m$ where
$\G_i^\prime$ contains output variables in $A_i$ and the variable
$D_i:A_i$.
A set of sequents $\mathset{\Sc{i} \sep 0 \leq i \leq m}$ is
defined where
$\Sc{i} = \mSeq {\Gc{i}} {\D_i} {\fBranch}$
\end{definition}
The following lemma relates this definition to mode checking:
\begin{lemma}\label{lem:const_ctx_relate_mode}
Given a clause $c$, the context $\Gc{i}$ as defined in Definition
\ref{def:const_trans} is a superset of the context $\G_i$ as defined
in Definition \ref{def:mode_csst} for $i$ such that $0 \leq i \leq m$.
\end{lemma}
\begin{proof}
It is easy to see that $\G_i$ can be obtained by removing variables
$D_j$ for $1 \leq j \leq i$ and variables that do not have a strict
occurrence in the type of $c$ from $\Gc{i}$.
\end{proof}
By this lemma and the definition of mode consistency, we know that for
a constant $c: A \from A_1 \from ... \from A_m$ the input variables to
$A_i$ are contained in $\Gc{i-1}$ for $1 \leq i \leq m$.

The output coverage goals and patterns in Definition \ref{def:out_cover} are
redefined as follows: the output coverage pattern for $A_i$ is now
$\lfSeq {\Gc{i-1}, \G^O_i} {\Si} {A_i} {\ktype}$
and the output coverage goal for
$A_i$ is now $\lfSeq {\Gc{i-1},\G^O[\si^I_i]} {\Si} {\lfAppTri a {\si^I_i}
  {\G^O}} {\ktype}$.
This is a superfluous change from Definition \ref{def:out_cover}: since the
context $\G^I_i$ containing input variables to $A_i$ is closed and
a subset of $\Gc{i-1}$, splitting and immediate coverage will never
affect variables that are not in $\G^I_i$.
Thus splitting and immediate coverage in Definition \ref{def:out_splitting} and
\ref{def:out_immed_cover} are adopted to this new definition
in a straightforward manner.

By the definition, $\Sc{0}$ is the sequent $S$.
We construct a proof for $S$ by recursively constructing proofs for
$\mathset{\Sc{i} \sep 0 \leq i \leq m}$.
In the base step we construct a proof for $\Sc{m}$.
In the recursive steps we construct proof for $\Sc{i-1}$ for $1 \leq i
\leq m$.

\textbf {Base Step.}
In the base step we translate the construction of outputs into an
application of \kexistsl to $\Sc{m}$ as shown in Figure
\ref{fig:proof_branch_base}.
\begin{figure}[h!]
\center
\begin{gather*}
  \infer[\kexistsl]{
    \mSeq {\Gc{m}} {\D_{m}} {\fBranch}
  }{
    \lfSeq {\Gc{m}} {\Si} {\si_{m}}
           {(\G^O[\siI], D: \lfAppTri a {\siI} {\G^O})}
  }
\end{gather*}
\nocaptionrule \caption{Finishing the Proof by Applying \kexistsr Rule}
\label{fig:proof_branch_base}
\end{figure}
By Definition \ref{def:const_trans}, $\Gc{m}$ contains $D_{1}:A_{1}, ...,
D_{m}:A_{m}$.
Letting $M_{k+1}, ..., M_n$ be the output arguments to the head of $c$
and $\siO = (M_{k+1}/x_{k+1}, ..., M_n/x_n)$, $\si_{m} = (\siO, (c\,
D_m\, ..., D_1)/D)$ in Figure \ref{fig:proof_branch_base}.
\begin{lemma}\label{lem:const_output}
$\lfSeq {\Gc{m}} {\Si} {\si_{m}}
  (\G^O[\siI], D: \lfAppTri a {\siI} {\G^O})$ holds.
\end{lemma}
\begin{proof}
By the definition of mode consistency and Lemma
\ref{lem:const_ctx_relate_mode},
$A$ is output consistent relative to $\Gc{m}$.
Thus $\Gc{m}$ contains all free variables in $\siO$. By the
well-typedness of $A$,
$\lfSeq {\Gc{m}} {\Si} {\siO} {\G^O[\siI]}$ holds.
Thus we have
\begin{gather*}
  \infer[\kwd{subst-typ}]{
    \lfSeq {\Gc{m}} {\Si}
           {\si_m}
           {(\G^O[\siI], D: \lfAppTri a {\siI} {\G^O})}
  }{
    \lfSeq {\Gc{m}} {\Si} {\siO} {\G^O[\siI]}
    &
    \lfSeq {\Gc{m}} {\Si} {\lfAppFour c {D_m} {...} {D_1}} {A}
  }
\end{gather*}
\end{proof}

\textbf {Recursive Steps. }
For $i$ such that $1 \leq i \leq m$, we create a proof tree for
$\Sc{i-1}$ by first translating the recursive call represented by
$A_i = \lfAppFour a {M_{i_1}} {...} {M_{i_n}}$ and then recursively
generating the proof for $\Sc{i}$.
A recursive call can be translated to an application of the inductive
hypothesis in the \M2 proof as shown in Figure
\ref{fig:proof_branch_app_ih}, in which
$\si^I_i = (M_{i_1}/x_1, ..., M_{i_k}/x_k)$
is a substitution containing input arguments to $A_i$.
\begin{figure}[h!]
\center
\begin{smallgather}
  \infer[\kfalll]{
    \mSeq {\Gc{i-1}} {\D_{i-1}} F
  }{
    \lfSeq {\Gc{i-1}} {\Si} {\si^I_i} {\G^I}
    &
    \infer[\kexistsl]{
      \mSeq {\Gc{i-1}}
            {\D_{i-1}^\prime}
            F
    }{
      \mSeq {\Gc{i-1}, \G^O[\si^I_i], D_i: \lfAppTri a {\si^I_i} {\G^O}}
            {\D_{i-1}^\prime}
            F
    }
  }
\\
\text{where
$\D_{i-1}^\prime = {\D_{i-1}, \mExst {\G^O[\si^I_i]} {\mExst {D_i:
      \lfAppTri a {\si^I_i} {\G^O}} {\mTrue}}}$}
\\
\text{and $F = \fBranch$}
\end{smallgather}
\nocaptionrule \caption{Applying the Inductive Hypothesis}
\label{fig:proof_branch_app_ih}
\end{figure}

In the simple case when the output arguments $M_{i_{k+1}} ... M_{i_n}$
to $A_i$ are variables, the frontier sequent resulting from applying
\kfalll and \kexistsl is exactly $\Sc{i}$.
By recursively generating the proof for $\Sc{i}$ we finish the translation.
However, in many cases the output arguments in $A_i$ are not
variables.
For example, in the \ksubred example, the first premise of \kevapp is
\begin{align*}
  & \ecdSubred
      {(Dev_1 : \ecdEval M {(\ecdAbs {M^\prime})})}
      {(Dty_1 : \ecdOf M {(\ecdArr {T_1} T)})}\\
  & \qquad
      (\ecdOfAbs {(Dty_3 :
        \lfTAbs {x} {\ktm} {\ecdOf x {T_1} \to \ecdOf {(\lfApp {M^\prime} x)} T})})
\end{align*}
which has an output $\kofabs~Dty_3$.
By applying the inductive hypothesis on $Dev_1$ and $Dty_1$, we get an
new variable
$D_3:\ecdOf {(\ecdAbs {M^\prime})} {(\ecdArr {T_1} T)}$ in the context.
We need to perform an inversion (case analysis) on $D_3$ to introduce
the output variable $Dty_3$ into the context.

The inversion consists of \kcase rules translated from the output
coverage checking.
The translation is similar to translating input coverage checking.
First we define a relation between output coverage goals and \M2 sequents.
\begin{definition}\label{def:roc}
Let $\roc$ be a binary relation between output coverage goals and \M2
sequents.
$G \roc S$ iff $G$ is an output coverage goal
$\lfSeq {\G^I, \G^O} {\Si} A {\ktype}$
where $\G^I$ contains input parameters and $\G^O$ contains output
parameters and $S$ is an \M2 sequent
$\mSeq {\G^I, \G^O, D:A} {\D} F$
where $\D$ contains the inductive hypothesis and
$F = \fBranch$.
\end{definition}
\begin{definition}
Given a finite set of output coverage goals
$\GSet = \mathset{G_i \sep 1 \leq i \leq n}$
and \M2 sequents
$\SSet = \mathset{S_i \sep 1 \leq i \leq n}$,
$\GSet \roc \SSet$
if and only if $G_i \roc S_i$ for $1\leq i \leq n$.
\end{definition}
Then the following lemma holds:
\begin{lemma}\label{lem:roc}
Given $G \roc S$ where $G$ is
$\lfSeq {\G^I, \G^O_1, x:A_x, \G^O_2} {\Si} A {\ktype}$
and $S$ is
$\mSeq {\G^I, \G^O_1, x:A_x, \G^O_2, D:A} {\D} F$,
if $\GSet$ is the set of subgoals resulting from splitting on $x$ in
$G$ and $\SSet$ is the set of sequents resulting from applying \kcase rule to
$x$ in $S$,
then $\GSet \roc \SSet$.
\end{lemma}
\begin{proof}
Similar to Lemma \ref{lem:rsc}. Note that all free variables in $F$
are bound in $\G^I$.
Since variables in $\G^I$ are not instantiated by splitting in output
coverage checking, they are also not instantiated by applying the
corresponding \kcase rule in the \M2 proof.
Thus $F$ remains the same after applying \kcase.
\end{proof}
By Lemma \ref{lem:roc}, we can prove a lemma similar to Lemma
\ref{lem:rscproof}, from which we derive an algorithm for translating
output coverage checking.
Let $\Sc{i-1}'$ be the frontier sequent resulting from applying
the inductive hypothesis (as shown in Figure
\ref{fig:proof_branch_app_ih}) and $G_{c_i}'$ be the output coverage
goal for $A_i$.
By the definition of $\roc$ and output coverage goals, we have
$G_{c_i}' \roc \Sc{i-1}'$.
Let $G_{c_i}$ be the output coverage patterns for $A_i$.
By Definition \ref{def:const_trans} and output coverage patterns, we have
$G_{c_i} \roc \Sc{i}$.
The output coverage checking performs a sequence of splitting
operations to $G_{c_i}'$ and returns a set of subgoal $\GSet'$
immediately covered by $G_{c_i}$.
Those splitting operations are translated into applications of \kcase
rules to the frontier sequent $\Sc{i-1}'$, resulting in a partial
proof with frontier sequents $\SSet'$ such that $\GSet' \roc \SSet'$.
To relate the frontier sequents in $\SSet'$ to $\Sc{i}$, we prove the
following lemma, which reflects the immediate coverage in output
coverage into equivalence between \M2 sequents up to renaming.
\begin{lemma}\label{lem:out_immed_cover_inst}
Given an output coverage goal $G^\prime$ and an output coverage
pattern $G$ such that $G^\prime$ is immediately covered by $G$, if $G
\roc S$ and $G^\prime \roc S^\prime$ then $S^\prime$ is the
instantiation of $S$ under $\si$ which is a substitution that only
renames variables.
\end{lemma}
\begin{proof}
Similar to Lemma \ref{lem:in_immed_cover_inst}.
Because we restricted the instantiation of immediate coverage to an
substitution that only renames variables in output coverage checking,
it is reflected into equivalence between \M2 sequents up to renaming.
\end{proof}
By applying Lemma \ref{lem:out_immed_cover_inst} to $\GSet' \roc
\SSet'$ and $G_{c_i} \roc \Sc{i}$, we have that for every $S' \in
\SSet'$, $S'$ is an instantiation of $\Sc{i}$ under a renaming
substitution.
The following lemma shows that instantiations of \M2 proofs under
renaming substitutions are also valid proofs.
\begin{lemma}\label{lem:proof_rename_inst}
Let $S'$ be an instantiation of $S$ under a renaming substitution, if
$S$ has a \M2 proof, then so does $S'$.
\end{lemma}
\begin{proof}
By induction on the structure of the proof for $S$ and case analysis
of the rule applied to $S$.
The proof for most of cases is straightforward.
When the rule applied is \kcase, the proof relies on the fact that for a
unification problem, compositions of mgus with renaming substitutions
are still mgus.
\end{proof}

At this point, we recursively generate a proof for $\Sc{i}$ and apply
Lemma \ref{lem:proof_rename_inst} to get proofs for sequents in $\SSet'$.
This finishes the construction of the proof for $\Sc{i-1}$.

\subsubsection{Correctness of the Translation of Clauses}
The following lemma shows that proofs translated from clauses are valid:
\begin{lemma}\label{lem:trans_clause_correct}
Given a clause
$c: A \from A_1 \from ... \from A_m$,
the input coverage pattern $G$ derived from $c$ and the \M2 sequent
$S$ such that $G \rsc S$.
The algorithm for translating clauses to \M2 proofs generate an \M2
proof for $S$ from $c$.
\end{lemma}
\begin{proof}
We define \M2 sequents
$\mathset{\Sc{i} \sep 0 \leq i \leq m}$ as in Definition \ref{def:const_trans}.
By the definition we have $S = \Sc{0}$.
We prove that the algorithm constructs a correct proof for $S$ by
induction on $i$.
In the base case, Lemma \ref{lem:const_output} guarantees the
\kexistsr is correctly applied to get a proof for $\Sc{m}$.
In the inductive case, assume we have a correct proof for $\Sc{i}$ for
some $i$ such that $1 \leq i \leq m$, we prove that applications of
$\kfalll$, $\kexistsl$ and $\kcase$ rules to $\Sc{i-1}$ are correct.
By the definition of mode consistency in Definition \ref{def:mode_csst} and
Lemma \ref{lem:const_ctx_relate_mode}, the context $\Gc{i-1}$ in the
premise of \kfalll,
$\lfSeq {\Gc{i-1}} {\Si} {\si^I_i} {\G^I}$,
contains all variables in $\si^I_i$.
Thus this premise holds by \LF typing rules and \kfalll is correctly
applied.
The correctness of \kexistsl and \kcase rules is obvious from their
definitions.
By Lemma \ref{lem:roc} and Lemma \ref{lem:out_immed_cover_inst}, after
applying those rules we get a partial proof with frontier sequents
$\SSet$, which is completed by applying Lemma \ref{lem:proof_rename_inst}
to the proof for $\Sc{i}$.
\end{proof}

\subsubsection{Correctness of the Instantiation of Proofs}\label{sec:inst_branch}
As described in the beginning of Section
\ref{sec:gen_proof_from_clauses}, we have to instantiate the proof
generated from clauses to derive frontier sequents in the partial
proof resulting from case analysis steps.
In general, the provability of an \M2 sequent is not closed under
instantiations.
For instance, a \kcase rule might no longer be applicable after an
instantiation because the side condition for \kcase, that every
unification problem must either has an mgu or no solution, may no
longer hold after the instantiation.
In our case, an \M2 proof to be instantiated is translated from a
totality checked clause, which has a particular structure such that the
same rules can be applied to prove the instantiation of the \M2 sequent.
Specifically, it does not use \krecur rule and \kcase rules in it are
translated from splitting in output coverage checking, which only
involves matching problems instead of general unification problems.
From those observations, we can prove the following instantiation lemma.
\begin{lemma}\label{lem:proof_branch_inst}
If $\drv$ is an \M2 proof generated from clauses as described in Section
\ref{sec:trans_proof_from_clauses} for the sequent
$S = \mSeq {\G} {\D} F$,
then given any substitution
$\lfSeq {\G^\prime} {\Si} {\si} {\G}$
and an instantiation $S^\prime$ of $S$ under $\si$ there exists a
proof $\drv^\prime$ for $S^\prime$.
\end{lemma}
\begin{proof}
The proof is by induction on the structure of $\drv$ and case analysis of
the rule applied to $S$ in $\drv$.
The important case is for the \kcase rule, in which we prove that the
structure of the \kcase are maintained under instantiation.
Proof for other cases is straightforward.
See Theorem \ref{lem:proof_branch_inst_full} in the appendix for the
details.
\end{proof}

\subsection{Correctness of the \M2 Proof Generation}\label{sec:correctness}
We have shown that almost all applications of the \M2 rules in our
proof generation are correct, except for \krecur rule in the initial
step which has a side condition that the proof term must represent a
terminating computation.
Given an totality checked \LF signature $\Si$ and a type family
$\exmTypeFamily$ in $\Si$ such that $\G^I = x_1:A_1, ..., x_k:A_k$, an
\M2 proof for
$\mSeq {\emptyCol} {\emptyCol} {\mTm P F}$
is generated by the translation algorithm, where $F = \exmFormula$ and
$P = \mRecur x F
      {\mIAll {(x_1:A_1, ..., x_k:A_k)} {P^\prime}}$.
Intuitively, the translation algorithm should reflect the termination
ordering in totality checking to that in the proof term $P$.
Then the termination property of $P$ is guaranteed by the termination
checking.
To formally prove that, we need an instantiation lemma for termination
ordering.
The following lemma proves that subterm ordering holds under instantiation.
\begin{lemma}\label{lem:termination_order_inst}
Let $\subtermOrder$ be the subterm ordering between \LF terms.
If $M \subtermOrder N$ and $\si$ is an substitution, then $M[\si]
\subtermOrder N[\si]$.
\end{lemma}
\begin{proof}
The formal definition of $\subtermOrder$ is described in
\cite{rohwedder96esop}.
The proof is by a straightforward induction on the derivation rules
for $\subtermOrder$.
\end{proof}
By this lemma, the same property can be proved easily for
lexicographical and simultaneous ordering.
These are all orderings supported by \M2.
Then we prove that the termination property holds for $P$.
\begin{lemma}\label{lem:recur_term}
If $P = \mRecur x F {\mIAll {(x_1:A_1, ..., x_k:A_k)} {P^\prime}}$ is
the proof term generated by the translation algorithm, then $P$
terminates in $x$.
\end{lemma}
\begin{proof}
Suppose every recursive call $\lfAppFour x {M_1} {...} {M_k}$ in the
proof term occurs in a position such that $x_1, ..., x_k$ are
instantiated by $\si=\substComp {\substComp {\si_1} {...}} {\si_m}$
where $\si_i$ for $1 \leq i \leq m$ are substitutions in \kcase rules
along the execution path and $\si$ is their compositions, we have to
prove that $M_1,...,M_k$ are smaller than $x_1[\si],...,x_k[\si]$
relative to some termination ordering.
In our translated proofs, a recursive call corresponds to a premise
$A$ in some clause $c$.
Let $M_1,...,M_k$ be the input arguments to $A$ and $M_1^\prime, ...,
M_k^\prime$ be the input arguments to the head of $c$.
By termination checking we have $M_1,...,M_k$ is smaller than
$M_1^\prime, ..., M_k^\prime$ relative to the termination ordering.
In the proof term $P'$, the recursive call occurs in a position where
$x_1, ..., x_n$ are instantiated by the \kcase rules translated from
input coverage checking.
Letting $\si$ be the composition of substitutions in those \kcase
rules, we have $x_i[\si] = M_i^\prime[\si^\prime]$ for $i$ such that
$1 \leq i \leq k$, where $\si^\prime$ is the substitution for
instantiating the proof translated from $c$, as described in
Section \ref{sec:gen_proof_from_clauses}.
By the translation algorithm, this recursive call has the form
$\lfAppFour x {(M_1[\si^\prime])} {...} {(M_k[\si^\prime])}$.
Applying Lemma \ref{lem:termination_order_inst} to the order between
$M_1,..,M_k$ and $M_1^\prime,..,M_k^\prime$, we get
${(M_1[\si^\prime])} {...} {(M_k[\si^\prime])}$ is smaller than
$x_1[\si], ..., x_k[\si]$ relative to the termination ordering.
\end{proof}

Then we have our main theorem:
\begin{theorem}
Given an \LF signature $\Si$ that is totality checked without using
contexts and lemmas,
for every type family
$a : \lfTAbsa {\G^I} {\lfTAbsa {\G^O} {\top}}$
in $\Si$, the proof generation algorithm generates an \M2 proof for
$\mSeq {\emptyCol} {\emptyCol}
       {\mTm P
             {\mFall {\G^I}
               {\mExst {\G^O}
                 {\mExst {D:\lfAppTri a {\G^I} {\G^O}} {\top}}}}}$.
\end{theorem}
\begin{proof}
We prove that every application of an \M2 rule is valid.
\begin{enumerate}
\item
Obviously, the application of \kfallr in the initial step is correct.
\item
The applications of \kcase rules translated from input coverage
checking are correct by Lemma \ref{lem:rscproof}.
\item
The translation from clauses to \M2 proofs are correct by
Lemma \ref{lem:trans_clause_correct}.
By instantiating those proofs the frontier sequents resulting
from translating input coverage checking are proved, which is proved
by Lemma \ref{lem:in_immed_cover_inst}.
The instantiated proofs are valid by Lemma \ref{lem:proof_branch_inst}.
\item
Finally, the application of \krecur rule in the initial step
is correct by Lemma \ref{lem:recur_term}.
\end{enumerate}
Since all the proof rules are correctly applied, we get
an \M2 proof for $\exmEndSequent$.
\end{proof}

\section{Conclusions and Future Works}\label{sec:conc}
This paper has described a method for transforming totality checking
in a restricted version of \Twelf into explicit proofs in the logic
\M2.
We have also proved this method correct.
Towards this end, we have adapted arguments for the correctness of
particular components of totality checking into an argument for the
correctness of the generated \M2 proof.
Since \M2 is a consistent formal logic with established
meta-properties, the existence of such a transformation boosts our
confidence in totality checking.
More importantly, the transformation yields explicit objects that can
be traded as proof certificates and whose correctness can be checked
independently of the procedure that generated them.

As we have noted already, we have not considered the entire class of
specifications for which Twelf supports totality checking.
One limitation is that we have not considered the proof of properties
that are parameterized by changing signatures that adhere to
constraints that are finitely described via regular worlds
descriptions in \Twelf.
Another limitation is that we have not allowed for clauses that
contain calls to predicates other than the predicate being defined.
We believe our work can be extended to treat totality checking in this
more general setting.
However, we will need a richer target logic than \M2.
In particular, the proofs we produce will have to be in the logic
\MTwoPlus that extends \M2 with judgments with \emph{generalized
  contexts} and that contains rules for lemma applications.
We plan to examine this issue more carefully in the near future.

Another direction of ongoing inquiry is the correspondence between the
formalization and validation of meta-theorems in the \LF/\Twelf framework
and in the two-level logic framework of \Abella \cite{gacek12jar}.
It has previously been shown that specifications in \LF can be
translated into clause definitions in the specification logic used in
Abella in a way that preserves essential structure \cite{snow10ppdp}.
By exploiting this correspondence, it appears possible to transform
\M2 (and possibly \MTwoPlus) proofs over \LF specifications into
\Abella proofs over the related \Abella specifications.
The results of this paper and its extension would then yield a path to
proofs in another tried and tested logical setup.

\medskip
\noindent%
\emph{Acknowledgments}:%
This work has been partially supported by the NSF Grant CCF-0917140.
Opinions, findings, and conclusions or recommendations expressed in this paper
are those of the authors and do not necessarily reflect the views of the
National Science Foundation.

\bibliographystyle{abbrv}
\bibliography{../../references/master}

\clearpage
\appendix

\section{Instantiation Theorem for \M2}

\begin{theorem}\label{lem:proof_branch_inst_full}
If $\drv$ is an \M2 proof generated from clauses as described in Section
\ref{sec:trans_proof_from_clauses} for the sequent
$S = \mSeq {\G_1, \G_2} {\D} F$
where $\G_1$ contains input parameters and $\G_2$ contains output
parameters, then given any substitution
$\lfSeq {\G_1^\prime} {\Si} {\si_1} {\G_1}$
there exists a proof for
$S^\prime = \mSeq{\G_1^\prime, \G_2[\si_1]} {\D[\si_1]} {F[\si_1]}$.
\end{theorem}
\begin{proof}
By induction on the structure of $\drv$.
According to the last rule used in $\drv$, there are the
following cases:
\noindent \textbf{Case:} \kfalll. 
The proof looks like
\begin{smallgather}
  \infer[\kfalll]{
    \mSeq {\G_1, \G_2} {\D_1, \mFall {\G_3} {F_3}, \D_2} {F_1}
  }{
    \lfSeq {\G_1, \G_2} {\Si} {\si} {\G_3}
    &
    \infer[]{
      \mSeq {\G_1, \G_2} {\D_1, \mFall {\G_3} {F_3}, \D_2, F_3[\si]} {F_1}
    }{
      \drv^{\prime}
    }
  }
\end{smallgather}
Applying I.H to $\drv^{\prime}$, we get
\begin{smallgather}
  \infer[]{ 
    \mSeq {\G_1^\prime, \G_2[\si_1]}
          {\D_1[\si_1], (\mFall {\G_3} {F_3})[\si_1], \D_2[\si_1], F_3[\si][\si_1]}
          {F_1[\si_1]}
  }{
    \drv^{\prime\prime}
  }
\end{smallgather}
Since $\mFall {\G_3} {F_3}$ is a closed term (an inductive hypothesis got from
applying \krecur), we have
$(\mFall {\G_3} {F_3})[\si_1] = \mFall {\G_3} {F_3}$.
The proof above is equivalent to
\begin{smallgather}
  \infer[]{ 
    \mSeq {\G_1^\prime, \G_2[\si_1]}
          {\D_1[\si_1], \mFall {\G_3} {F_3}, \D_2[\si_1], F_3[\si][\si_1]}
          {F_1[\si_1]}
  }{
    \drv^{\prime\prime}
  }
\end{smallgather}
Let $\si_1^\prime = (\si_1, \si_{id})$ where $\si_{id}$ is an identity
substitution such that
$\lfSeq {\G_1^\prime, \G_2[\si_1]} {\Si} {\si_1^\prime} {\G_1, \G_2}$.
Applying Lemma. \ref{lem:subst_comp} to 
$\lfSeq {\G_1^\prime, \G_2[\si_1]} {\Si} {\si_1^\prime} {\G_1, \G_2}$
and
$\lfSeq {\G_1, \G_2} {\Si} {\si} {\G_3}$, 
we get
$\lfSeq {\G_1^\prime, \G_2[\si_1]} {\Si} {\substComp {\si} {\si_1^\prime}} {\G_3[\si_1^\prime]}$.
Since $\G_3$ is closed,  $\G_3[\si_1^\prime] = \G_3$.
We apply \kfalll to conclude this case, where
$\drv^{\prime\prime\prime}$ is the derivation for
$\lfSeq {\G_1^\prime, \G_2[\si_1]} {\Si} {\substComp {\si}{\si_1^\prime}} {\G_3}$:
\begin{smallgather}
  \infer[\kfalll]{
    \mSeq {\G_1^\prime, \G_2[\si_1]}
          {\D_1[\si_1], \forall \G_3.F_3, \D_2[\si_1]}
          {F_1[\si_1]}
  }{
    \drv^{\prime\prime\prime}
    &
    \infer[]{
      \mSeq {\G_1^\prime, \G_2[\si_1]} 
            {\D_1[\si_1], \forall \G_3.F_3, \D_2[\si_1], F_3[\si][\si_1]}
            {F_1[\si_1]}
    }{
      \drv^{\prime\prime}
    }
  }
\end{smallgather}

\noindent \textbf{Case:} \kfallr.
The proof looks like:
\begin{smallgather}
  \infer[\kfallr]{
    \mSeq {\G_1,\G_2} {\D} {\mFall {\G_3} F}
  }{
    \infer[]{
      \mSeq {\G_1,\G_2,\G_3} {\D} F
    }{
      \drv^{\prime}
    }
  }
\end{smallgather}
By I.H, we have
\begin{smallgather}
  \infer[]{
    \mSeq {\G_1^\prime, \G_2[\si_1], \G_3[\si_1]} {\D[\si_1]} {F[\si_1]}
  }{
    \drv^{\prime\prime}
  }
\end{smallgather}
We conclude this case by applying \kfallr:
\begin{smallgather}
  \infer[\kfallr]{
    \mSeq {\G_1^\prime, \G_2[\si_1]} {\D[\si_1]} {(\mFall {\G_3} F)[\si_1]}
  }{
    \infer[]{
      \mSeq {\G_1^\prime, \G_2[\si_1], \G_3[\si_1]} {\D[\si_1]} {F[\si_1]}
    }{
      \drv^{\prime\prime}
    }
  }
\end{smallgather}

\noindent \textbf{Case:} \kexistsl.
The proof looks like:
\begin{smallgather}
  \infer[\kexistsl]{
    \mSeq {\G_1, \G_2} {\D_1, \mExst {\G_3} {\top}, \D_2} F
  }{
    \infer[]{
      \mSeq {\G_1, \G_2, \G_3} {\D_1, \mExst {\G_3} {\top}, \D_2} F
    }{
      \drv^{\prime}
    }
  }
\end{smallgather}
By I.H., we have
\begin{smallgather}
  \infer[]{
    \mSeq {\G_1^\prime, \G_2[\si_1], \G_3[\si_1]}
          {\D_1[\si_1], \mExst {\G_3[\si_1]} {\top}, \D_2[\si_1]}
          {F[\si_1]}
  }{
    \drv^{\prime\prime}
  }
\end{smallgather}
We apply $\kexistsl$ to conclude this case:
\begin{smallgather}
  \infer[\kexistsl]{
    \mSeq {\G_1^\prime, \G_2[\si_1]}
          {\D_1[\si_1], \mExst {\G_3[\si_1]} {\top}, \D_2[\si_1]}
          {F[\si_1]}
  }{
    \infer[]{
      \mSeq {\G_1^\prime, \G_2[\si_1], \G_3[\si_1]}
            {\D_1[\si_1], \mExst {\G_3[\si_1]} {\top}, \D_2[\si_1]}
            {F[\si_1]}
    }{
      \drv^{\prime\prime}
    }
  }
\end{smallgather}

\noindent \textbf{Case:} \kexistsr.
The proof looks like:
\begin{smallgather}
  \infer[\kexistsr]{
    \G_1,\G_2; \D \m2seq \exists \G_3.\top
  }{
    \G_1, \G_2 \lfseq \si : \G_3
  }
\end{smallgather}
Let $\si_1^\prime = (\si_1, \si_{id})$ where $\si_{id}$ is an identity
substitution such that
$\lfSeq {\G_1^\prime, \G_2[\si_1]} {\Si} {\si_1^\prime} {\G_1, \G_2}$.
Applying Lemma. \ref{lem:subst_comp} to 
$\lfSeq {\G_1^\prime, \G_2[\si_1]} {\Si} {\si_1^\prime} {\G_1, \G_2}$
and
$\lfSeq {\G_1, \G_2} {\Si} {\si} {\G_3}$, 
we have
$\lfSeq {\G_1^\prime, \G_2[\si_1]} {\Si} {\substComp {\si} {\si_1}} {\G_3[\si_1^\prime]}$.
Note that $\G_3[\si_1^\prime] = \G_3[\si_1]$.
We apply \kexistsr to conclude this case:
\begin{smallgather}
  \infer[\kexistsl]{
    \G_1^\prime, \G_2[\si_1]; \D[\si_1] \m2seq \exists \G_3[\si_1]. \top
  }{
    \G_1^\prime, \G_2[\si_1] \lfseq \si[\si_1] : \G_3[\si_1]
  }
\end{smallgather}

\noindent \textbf{Case:} \kcase. 
This is the only non-trivial case.
The proof looks like:
\begin{smallgather}
  \infer[\kcase]{
    \mSeq {\G_1, x:A_x, \G_2} {\D} {F}
  }{
    \infer[]{
      \mcSeq {\G_1, x:A_x, \G_2} {\D} {\Si} {F}
    }{
      \drv^{\prime}
    }
  }
\end{smallgather}
Since $\si_1$ only instantiate input variables and \kcase rules are applied
only on output variables, $x$ and variables in $\G_2$ will not be instantiated
by $\si_1$ (but their types will).
For a constant $c:\lfTAbsa {\G_c} {A_c}$ in $\Si$, we solve the
unification problem $(A_c = A_x, x = \lfApp c {\G_c})$.
If it does not have a solution, \ksignonuni is applied.
If it has an mgu, \ksiguni is applied.
If it has an unifier but not an mgu, \kcase rule is not applicable.
Thus in a correct proof such as $\drv$, the last situation does rise.
We prove that after the instantiation, the same rule can be applied to
get a correct application of \kcase rule.

The application of case rule after unification looks like:
\begin{smallgather}
  \infer[\kcase]{
    \mSeq {\G_1^\prime, x:A_x[\si_1], \G_2[\si_1]} {\D[\si_1]} {F[\si_1]}
  }{
    \mcSeq {\G_1^\prime, x:A_x[\si_1], \G_2[\si_1]} {\D[\si_1]} {\Si} {F[\si_1]}
  }
\end{smallgather}
For a constant $c:\lfTAbsa {\G_c} {A_c}$ in $\Si$, the unification
problem becomes $(A_x[\si_1] = A_c, x = \lfApp c {\G_c})$, which is
equivalent to $(A_x = A_c, x = \lfApp c {\G_c})[\si_1]$ since $\si_1$
does not contain variables in $\lfTAbsa {\G_c} {A_c}$.
If $(A_x = A_c, x = \lfApp c {\G_c})$ does not have a solution,
so does $(A_x[\si_1] = A_c, x = \lfApp c {\G_c})$.
Thus we can apply the same \ksignonuni rule in this case.
If $(A_x = A_c, x = \lfApp c {\G_c})$ has an mgu $\si$, according to
the translation algorithm the original proof looks like:
\begin{smallgather}
\infer[\ksiguni]{
  \mcSeq
      {\G_1, x:A_x, \G_2}
      {\D}
      {\Si,c:\lfTAbsa {\G_c} {A_c}}
      F
}{
  \infer[]{
    \mSeq {\G^\prime, \G_2[\sigma]} {\D[\sigma]} {F[\sigma]}
  }{
    \drv_1
  }
  &
  \infer[]{
    \mcSeq {\G_1,x:A_x,\G_2} {\D} {\Si} F
  }{
    \drv_2
  }
}
\\
\lfSeq {\G^\prime} {} {\sigma} {(\G_1, x:A_x, \G_c)} \text{ where }
\sigma = mgu(A_x = A_c, x= c\, \G_c)  
\end{smallgather}
Because this is translated from splitting in output coverage checking,
we have $\si = (\si_{id}, (\lfApp c {\G_c})[\si_c]/x, \si_c)$ and
$\G^\prime = \G_1, \G_c^\prime$ according to
Def. \ref{def:out_splitting}, where $\si_{id}$ is an identity
substitution for $\G_1$ and $\si_c$ is the substitution for $\G_c$.
Since $\G_c$ is disjoint from $\G_1$, we have $\lfSeq {\G_1, \G_c'}
{\Si} {\si_c} {\G_c}$.
The proof above is equivalent to 
\begin{smallgather}
\infer[\ksiguni]{
  \mcSeq
      {\G_1, x:A_x, \G_2}
      {\D}
      {\Si,c:\lfTAbsa {\G_c} {A_c}}
      F
}{
  \infer[]{
    \mSeq {\G_1, \G_c^\prime, \G_2[\si]} {\D} {F}
  }{
    \drv_1
  }
  &
  \infer[]{
    \mcSeq {\G_1,x:A_x,\G_2} {\D} {\Si} F
  }{
    \drv_2
  }
}  
\end{smallgather}
By applying I.H. to $\drv_1$, we get the following proof:
\begin{smallgather}
\infer[]{
  \mSeq {\G_1^\prime, \G_c', \G_2[\si][\si_1]} {\D[\si_1]} {F[\si_1]}
}{
  \drv^{\prime\prime}
}
\end{smallgather}
Let $\si_1' = (\si_1, \si_{id})$ such that
$\lfSeq {\G_1', \G_c'} {\Si} {\si_1'} {(\G_1, \G_c')}$,
we have $\lfSeq {\G_1', \G_c'} {\Si} {\substComp {\si_c} {\si_1'}} {\G_c}$.
Let 
\begin{smallgather}
\si' = (\si_{id}, (\lfApp c {\G_c})[\substComp {\si_c}{\si_1'}]/x, \substComp {\si_c} {\si_1'})
\end{smallgather}
such that $\lfSeq {\G_1', \G_c'} {\Si} {\si'} {(\G_1', x:A_x[\si_1], \G_c)}$.
If $\si^\prime$ is an mgu for $(A_x[\si_1] = A_c, x = \lfApp c
{\G_c})$, then we get a proof after instantiation using \ksiguni rule:
\begin{smallgather}
\infer[\ksiguni]{
  \mcSeq
      {\G_1^\prime, x:A_x[\si_1], \G_2[\si_1]}
      {\D[\si_1]}
      {\Si,c:\lfTAbsa {\G_c} {A_c}}
      F[\si_1]
}{
  \infer[]{
    \mSeq {\G_1', \G_c', \G_2[\sigma_1][\si']} {\D[\sigma_1]} {F[\sigma_1]}
  }{
    \drv^{\prime\prime}
  }
  &
  \drv_2^\prime
}  
\end{smallgather}
where $\drv_2^\prime$ is the rest of the instantiated proof.
Here we use the fact that $\G_2[\si][\si_1] = \G_2[(\si_1, (\lfApp c
  {\G_c})[\substComp {\si_c}{\si_1'}]/x, \substComp {\si_c} {\si_1'})]
= \G_2[\si_1][\si']$.

We then prove that $\si^\prime$ is indeed an mgu for $(A_x[\si_1] =
A_c, x = \lfApp c {\G_c})$.
\begin{enumerate}
\item
$\si^\prime$ is a unifier.
Since $\si$ is an mgu for $(A_x = A_c, x = \lfApp c {\G_c})$, we have
$A_x[\si] = A_c[\si]$ which is equivalent to $A_x = A_c[\si_c]$.
Thus $A_x[\si_1][\si^\prime] = A_x[\si_1] = A_c[\si_c][\si_1] = A_c[\substComp {\si_c} {\si_1'}] = A_c[\si^\prime]$.
Also $x[\si^\prime] = (\lfApp c {\G_c})[\substComp {\si_c} {\si_1'}] = (\lfApp c
{\G_c})[\si^\prime]$.

\item
$\si'$ is an mgu.
Let $\si_2 = (\si_{1_2}, M/x, \si_{c_2})$ be a unifier for
$(A_x[\si_1] = A_c, x = \lfApp c {\G_c})$ such that $\lfSeq {\G_2}
{\Si} {\si_2} {(\G_1', x:A_x[\si_1], \G_c)}$.
Let $\si_1'' = (\si_1, \si_{id})$ such that $\lfSeq {\G_1',
  x:A_x[\si_1], \G_c} {\Si} {\si_1''} {(\G_1, x:A_x, \G_c)}$.
We have $\substComp {\si_1''} {\si_2}$ as a unifer for $(A_x = A_c, x =
\lfApp c {\G_c})$ such that $\lfSeq {\G_2} {\Si} {\substComp {\si_1''}
  {\si_2}} {(\G_1, x:A_x, \G_c)}$.
Since $\si$ is an mgu for $(A_x = A_c, x = \lfApp c {\G_c})$, there
exists some $\si_3 = (\si_{1_3}, \si_{c_3})$ such that $\lfSeq {\G_2}
{\Si} {\si_3} {(\G_1, \G_c')}$ and $\substComp {\si} {\si_3} =
\substComp {\si_1''} {\si_2}$.
From $\substComp {\si} {\si_3} = \substComp {\si_1''} {\si_2}$, we
have $\si_{1_3} = \substComp {\si_1} {\si_{1_2}}$ and $\substComp
{\si_c} {\si_{3}} = \si_{c_2}$.
Letting $\si_4 = (\si_{1_2}, \si_{c_3})$, we prove that $\si_2 =
\substComp {\si'} {\si_4}$.
For this we need to prove the following:
\begin{itemize}
\item
$\si_{1_2} = \substComp {\si_{id}} {\si_4}$. This is obvious.
\item 
$\si_{c_2} = \substComp {(\substComp {\si_c} {\si_1'})} {\si_4}$.
We have $\substComp {\si_1'} {\si_4} = (\substComp {\si_1} {\si_{1_2}}, \si_{c_3}) 
= (\si_{1_3}, \si_{c_3}) = \si_3$. 
Thus $\si_{c_2} = \substComp {\si_c} {\si_{3}} = \substComp {\si_c} {(\substComp {\si_1'} {\si_4})}
= \substComp {(\substComp {\si_c} {\si_1'})} {\si_4}$.
\item
$M = (\lfApp c {\G_c})[\substComp {\si_c} {\si_1'}][\si_4]$.
Since $x[\si_2] = (\lfApp c {\G_c})[\si_2]$, i.e., $M = (\lfApp c {\G_c})[\si_2]$, we have
$M = (\lfApp c {\G_c})[\si_{c_2}] = (\lfApp c {\G_c})[\substComp {(\substComp {\si_c} {\si_1'})} {\si_4}]
= (\lfApp c {\G_c})[\substComp {\si_c} {\si_1'}][\si_4]$
\end{itemize}
Thus, for any unifier $\si_2$ for $(A_x[\si_2] = A_c, x = \lfApp c
{\G_c})$ there exists a substitution $\si_4$ such that $\si_2 = \substComp {\si'} {\si_4}$.
We conclude that $\si'$ is an mgu for $(A_x[\si_2] = A_c, x = \lfApp c {\G_c})$.
\end{enumerate}

\end{proof}

\end{document}